\documentclass[a4paper,thm-restate]{lipics-v2021}
\usepackage{hyperref}
\usepackage{amsmath,amsfonts,amssymb,xspace,mathtools, graphicx, amssymb}
\usepackage{graphicx}
\usepackage{hyperref}
\usepackage{xcolor}
\usepackage{thm-restate}
\usepackage{tkz-euclide}
\usepackage{tikz}
\usepackage{float}
\usepackage{multirow, tabularx}
\usepackage{tabularray}

\def\DTW{d_{\text{DTW}}}
\def\eps{\varepsilon}
\def\RR{{\mathbb R}}
\def\NN{{\mathbb N}}
\def\kinst{{X'}}
\def\dbainst{{X}}
\def\kassign{{\pi'}}
\def\dbaassign{{\pi}}
\def\positions{{\mathcal{P}}}
\def\centers{{\mathcal{T}}}
\def\curveone{{\gamma_1}}
\def\curvetwo{{\gamma_2}}

\def\Warping{{\mathcal W}}
\def\assignment{{\mathcal A}}

\DeclareMathOperator*{\argmin}{arg\,min}

\DeclarePairedDelimiter\abs{\lvert}{\rvert}
\DeclarePairedDelimiter\norm{\lVert}{\rVert}

\renewcommand{\emph}[1]{\textit{\textbf{#1}}}

\title{On the number of iterations of the DBA algorithm}

\nolinenumbers

\author{Frederik {Br\"uning}}{Department of Computer Science, University of Bonn, Germany}{}{}{}

\author{Anne Driemel}{Hausdorff Center for Mathematics, University of Bonn, Germany}{}{}{}

\author{Alperen Erg\"ur}{The University of Texas at San Antonio, partially supported by NSF CCF 2110075}{}{}{}

\author{Heiko R\"oglin}{Department of Computer Science, University of Bonn, Germany}{}{}{}

\authorrunning{F. Br\"uning, A. Driemel, A. Erg\"ur and H. R\"oglin}



\Copyright{Frederik {Br\"uning}, Anne Driemel, Alperen Erg\"ur and Heiko R\"oglin}

\begin{document}
\hideLIPIcs
\maketitle

\begin{abstract}
    The DTW Barycenter Averaging (DBA) algorithm is a widely used algorithm for estimating the mean of a given set of point sequences. In this context, the mean is defined as a point sequence that minimises the sum of dynamic time warping distances (DTW).
The algorithm is similar to the $k$-means algorithm in the sense that it alternately repeats two steps: (1)~computing an optimal assignment to the points of the current mean, and (2)~computing an optimal mean under the current assignment.
The popularity of DBA can be attributed to the fact that it works well in practice, despite any theoretical guarantees to be known.
In our paper, we aim to initiate a theoretical study of the number of iterations that DBA performs until convergence. We assume the algorithm is given $n$ sequences of $m$ points in $\RR^d$ and a parameter $k$ that specifies the length of the mean sequence to be computed. We show that, in contrast to its fast running time in practice, the number of iterations can be exponential in $k$ in the worst case --- even if the number of input sequences is $n=2$. We complement these findings with experiments on real-world data that suggest this worst-case behaviour is likely degenerate.
To better understand the performance of the algorithm on non-degenerate input, we study DBA in the model of smoothed analysis, upper-bounding the expected number of iterations in the worst case under random perturbations of the input. Our smoothed upper bound is polynomial in $k$, $n$ and $d$, and for constant $n$, it is also polynomial in $m$.
For our analysis, we adapt the set of techniques that were developed for analysing $k$-means and observe that this set of techniques is not sufficient to obtain tight bounds for general $n$.
\end{abstract}


\section{Introduction}
The DTW Barycenter Averaging (DBA) algorithm~\cite{petitjean2011} was introduced by Petitjean, Ketterlin and Gan{\c c}arski in 2011 and has since been used in many applications for clustering time series data. The objective is to find a representative time series that optimally summarizes a given set of input time series. Here, the optimality is measured using the sum of DTW distances to the input sequences, in the sense that, the smaller these distances are, the more fitting the representation is. 
The dynamic time warping distance (DTW) was introduced by Sakoe and Chiba \cite{Sakoe78} in the 1970s to capture similarities of time series in speech recognition and belongs to the family of elastic distance measures. It was later rediscovered by Berndt and Clifford~\cite{berndt1994using} and popularized by Keogh and Ratanamahatana~\cite{keogh2005exact}. By now it has long been a staple in the area of time series classification and is often used as a baseline comparison~\cite{lines2015time}.
DTW has been used in a variety of applications from signature verification and touch screen authentication to gait analysis and classification of surgical processes~\cite{ DeLuca2012, forestier2012classification, Gavrila95, Shanker2007}. 
The DBA algorithm has found application in the context of optimization of energy systems~\cite{teichgraeber2019clustering}, forecasting of household electricity demand~\cite{teeraratkul2017shape}, and human activity recognition~\cite{seto2015multivariate} to name a few. The representation computed by DBA is used to speed up classification tasks on time series~\cite{petitjean2016faster} and to improve the training of neural networks~\cite{fawaz2018data,forestier2017generating}.
The original paper by Petitjean et al.~\cite{petitjean2011} that introduced DBA has more than 800 citations giving witness to the fact that DBA is widely successful and popular across different application areas.

DBA consists of two update steps that are repeated in a loop after initialization of the representative sequence. The first step computes an optimal warping path (refer to Section~\ref{sec:prelims} for a precise description)  for each  input sequence with respect to current representative sequence. These warping paths induce an order-preserving assignment of the points of the input sequences to the points of the current representative sequence. In the second update step, each point of the representative sequence is updated with the mean of the group of points assigned to it. This process is repeated until the algorithm converges to a stable solution.
Despite its popularity, DBA is a heuristic in the double meaning that it neither comes with any guarantees on the quality of the solution nor on the running time. It may converge in a local minimum that is far from the global optimum that minimizes the sum of distances. 
Historically, this popularity may be due to the fact that the underlying optimization problem it tries to solve, which is sometimes referred to as DTW-MEAN, is NP-hard \cite{buchin2020hardness,bulteau2020tight} and there are no efficient approximization algorithms known. 
We briefly summarize what is known on the computational complexity. Bulteau et al. show that the problem is also not fixed-parameter tractable in the number of sequences~\cite{bulteau2020tight}.
Using dynamic programming one can find a time series that minimizes the sum of the DTW distances to $n$ time series in $O(m^{2n+1}2^n n)$ time, as shown by Brill et al.~\cite{brill2019}, where $m$ denotes an upper bound on the length of each time series. This approach readily extends to point sequences in higher dimensions leading to a running time of  $O(d m^{2n+1}2^n n)$.
Using an implementation of their exact algorithm, Brill et al.~empirically study properties like uniqueness and length of the representative sequence and discuss the  effect of fixing the length in advance, as it is done in DBA.
They conclude that ``choosing a mean length that is larger than the optimal value induces only a small structural error in most cases.''
Buchin et al.~\cite{buchin2022} consider the problem variant where the set of feasible solutions is restricted to sequences of a given length $k$ or shorter. They call this the length-restricted mean under DTW and present an exact algorithm with running time  $O(n^{2dk}m^{8dk^2})$, where the input are $n$ sequences of $m$ points in $\RR^d$.

\subparagraph{Objective}
In this paper, we initiate a study of the number of iterations that DBA performs until convergence. While convergence properties of the algorithm have been studied before~\cite{schultz2018nonsmooth}, it seems that the running time of the algorithm has not been the subject of rigorous study up to now.
We follow a line of thought that has proved successful for a closely related algorithm---the well-known $k$-means algorithm by Stuart Lloyd \cite{lloyd1982least}. For this algorithm it is known that the number of iterations in the worst case is exponential \cite{arthur2006slow,harpeled2005, inaba2000variance, vattani2011}. However, on most practical instances the $k$-means algorithm is reported to converge very fast. Moreover, using smoothed analysis, it has been shown  that the expected running time under random perturbations of the input is merely polynomial~\cite{roeglin2011}. This raises the question, to which extent these techniques may be applied in the analysis of DBA.

\subparagraph{Overview}
In Section~\ref{sec:upperbounds} we give two different upper bounds for the number of iterations of DBA in the worst-case. The first one is an exponential upper bound that is based on  techniques from real algebraic geometry and uses specific properties of the space of dynamic time warping paths. 
The second bound is based on a potential function argument and it depends on certain geometric properties of the input data. More precisely, we obtain a  worst-case upper bound that is linear in the length of the point sequences, and linear in $\frac{1}{\varepsilon}$. Here $\varepsilon$ is the minimal distance between any two mean point sequences that may be visited in the iterations of DBA.

In Section~\ref{sec:smoothedanalysis} we present our upper bounds for the expected number of iterations in a semi-random data model. More precisely, we perform smoothed analysis of the number of iterations of DBA under Gaussian perturbation (with any variance $\sigma^2$) of deterministic data. The techniques we use for the smoothed analysis are quite versatile, i.e. anti-concentration estimates and standard tail bounds for the norm of a random vector. So the obtained results can be easily generalised to broader distributions, e.g. sub-gaussian random variables. However, we prefer to present the ideas in a less technical manner on Gaussian perturbations.
We show that the expected number of iterations
until DBA converges is at most 
$ \widetilde{O} \left(n^2 m^{8\frac{n}{d}+6}d^4k^6\sigma^{-2} \right)$, where  the $\widetilde{O}(\cdot)$-notation omits logarithmic factors.

These upper bounds are complemented by an exponential worst-case lower bound in Section~\ref{sec:lowerbound}. In particular, we show that there is an instance of two point sequences with length $m=\Theta(k)$ in the plane such that DBA needs $2^{\Omega(k)}$ iterations to converge. 
The techniques in this section borrow from earlier work of Vattani \cite{vattani2011}. Interestingly, our lower bound shows this behaviour already for only $n=2$ sequences. In this setting, there also exists a dynamic programming algorithm that solves the same problem in $O(k^5)$ time~\cite{brill2019}, as mentioned above.
Furthermore, we observe in Section~\ref{sec:experiments} that, when applied to real-world data, the number of iterations that DBA performs on average is much lower than our theoretical analysis suggests. In particular, we observe only sublinear dependencies on any of the parameters $n, m$ and $k$. 
These empirical results support the implicit assumptions in our theoretical study of the algorithm, as the smoothed analysis under randomly perturbed instances avoids artificially constructed boundary cases and corresponding artificially high lower bounds.

\subsection{Preliminaries}\label{sec:prelims}
For $n\in \NN$, we define $[n]$ as the set $\{1,\dots,n\}$. We call an ordered sequence of points $p_1,\dots,p_m$ in $\RR^d$ a \emph{point sequence} of length $m$.
For two points $p,q$ in  $\RR^d$, we denote with  $\norm{p-q}^2$ the \emph{squared Euclidean distance}, where $\norm{.}$ is the standard Euclidean norm.  For $m_1,m_2\in\NN$, each sequence
$(1, 1) = (i_1, j_1), (i_2, j_2), \dots , (i_M, j_M) = (m_1, m_2)$
such that $i_k -i_{k-1}$ and $j_k -j_{k-1}$ are either $0$ or $1$ for all $k$ is a \emph{warping path} from $(1, 1)$ to $(m_1, m_2)$. We denote with $\Warping_{m_1,m_2}$  the set of
all warping paths from $(1, 1)$ to $(m_1, m_2)$. For any two point sequences $\gamma_1=(\gamma_{1,1},\ldots,\gamma_{1,m_1})\in(\RR^d)^{m_1}$ and $ \gamma_2=(\gamma_{2,1},\ldots,\gamma_{2,m_2})\in(\RR^d)^{m_2}$, we also write $\Warping_{\gamma_1,\gamma_2} = \Warping_{m_1,m_2}$, and call elements of $\Warping_{\gamma_1,\gamma_2}$
warping paths between $\gamma_1$ and $\gamma_2$. The \emph{dynamic time warping distance} between
the sequences $\gamma_1$ and $\gamma_2$ is defined as
\[\DTW(\gamma_1, \gamma_2) = \min_{w\in\Warping_{\gamma_1,\gamma_2}}\sum_{(i,j)\in w}\norm{\gamma_{1,i}-\gamma_{2,j}}^2\]
A warping path  that attains the above minimum is called an \emph{optimal warping
path} between $\gamma_1$ and $\gamma_2$. We denote with $\Warping_{m_1,m_2}^*\subset \Warping_{m_1,m_2}$ the set of warping paths $w$ such that there exist point sequences $\gamma_1\in(\RR^d)^{m_1}$ and $\gamma_2\in(\RR^d)^{m_2}$ with optimum warping path $w$. 
Let $X=\{\gamma_1,\dots, \gamma_n\}\subset (\RR^d)^{m}$ be a set of $n$ point sequences  of length $m$ and $C\in (\RR^d)^{k}$ be a point sequence of length $k$.
We call a sequence $\pi$  of tuples $(p,c)$ where $p$ is an element of some point sequence in $X$ and $c$ is an element of the point sequence $C$ an \emph{assignment map} between $X$ and $C$. 
We call an assignment map between $X$ and $C$ \emph{valid} if for each $1\leq i\leq n$  the sequence of all tuples $(p,c)$ of $\pi$ for which $p$ is a point of $\gamma_i$ forms a warping path $w(\pi)_i$ between $\gamma_i$ and $C$. We call a valid assignment map \emph{optimal} if for each $1\leq i\leq n$ the formed warping path is an optimal warping path. We define the \emph{cost} of an \emph{assignment map} $\pi$ as
$\Phi(\pi)=\sum_{(p,c)\in \pi}\norm{p-c}^2$. Similarly, we define the \emph{total warping distance} of a valid assignment map $\pi$ with respect to a point sequence $x=(x_1,\dots,x_k)$ as $\Psi_{\pi}(x) = \sum_{i=1}^{n}  \sum_{(j_1,j_2)\in w(\pi)_i}\norm{\gamma_{i,j_1}-x_{j_2}}^2$.

\subsection{The DBA Algorithm}\label{sec:dbadef}

Let $X$ be a set of $n$ point sequences $\gamma_1,\dots,\gamma_n\in(\RR^d)^m$.
Let $C\in(\RR^d)^k$ be another point sequence and $w^{(1)},\dots,w^{(n)}\in \Warping_{m,k}$ be chosen such that $w^{(i)}$ is an optimal warping path between $\gamma_i$ and  $C$. Then $w^{(1)},\dots,w^{(n)}$ define an assignment map $\pi$ between $X$ and  $C$. The assignment map $\pi$ can be represented by sets $S_1(\pi),\ldots,S_k(\pi)$, where 
\[S_i(\pi)= \cup_{j=1}^n\{\gamma_{j,t}\mid (i,t)\in w^{(j)}\}.\] 
 By construction, $\pi$ minimizes the DTW distances between $\gamma_1,\dots,\gamma_n$ and $C$.
In the opposite direction, the following sequence $C_\pi$ minimizes the sum of squared distances for fixed $\pi$.

\[ C_\pi =\left( c_1(\pi),c_2(\pi), c_3(\pi), \ldots, c_{k}(\pi)  \right)\]  where $c_i(\pi) := \frac{1}{\abs{S_i(\pi)}} \sum_{p \in S_i(\pi)} p. $
DBA alternately computes such assignment maps and  average point sequences as follows.

\begin{enumerate}
    \item Let $\pi_0$ be an initial assignment map (e.g. randomly drawn $w_0^{(1)},\dots,w_0^{(n)}$). Let $j \gets 0$.
    \item Let $j \gets j+1$. Compute the average point sequence $C_{\pi_{j-1}}$ based on $\pi_{j-1}$.
    \item Compute optimal warping paths $w^{(i)}_{j}$ between $\gamma_i$ and $C_{\pi_{j-1}}$ for all $1\leq i\leq n$. The warping paths define  an optimal assignment map $\pi_j$ between $X$ and $C_{\pi_{j-1}}$.
    \item If $\Phi(\pi_{j}) \neq \Phi(\pi_{j-1}) $, then go to Step 2. Otherwise, terminate.
\end{enumerate}

\section{Upper bounds}\label{sec:upperbounds}

We present two different upper bounds on the number of steps performed by the DBA algorithm. The first approach is based on a theorem from real algebraic geometry and the resulting bound holds for any input data.   The second approach is based on a potential function argument and it uses a geometric assumption on the input data. 

\subsection{An unconditional upper bound} \label{algebraic}

In this section we will employ some classical tools from real algebraic geometry to derive an upper bound on the number of steps performed by DBA. 
We recall a special case of a general real algebraic geometry fact suited for our purposes, see \cite{basu} for a more recent updated version of this result. First, we need some definitions.

\begin{definition}
A semi-algebraic set in $\mathbb{R}^n$ is a finite union of the sets of the form
\[  \{ x \in \mathbb{R}^n : f_1(x), \dots, f_\ell(x)=0 , g_1(x),\dots, g_k(x) > 0\} \]
where $f_i$ and $g_j$ are polynomials with real coefficients. A semi-algebraic set $X \subset \mathbb{R}^n$ is semi-algebraically connected if for any semi-algebraic sets $X_1$ and $X_2$ that are closed in $X$, that are disjoint, and $X=X_1 \cup X_2$, then we have $X=X_1$ or $X=X_2$.
\end{definition}
The prototypical example of a  semi-algebraically connected set is the open cube $(0,1)^n$, and every semi-algebraic set can be decomposed into semi-algebraically connected components. The following is a special case of the classical Oleinik-Petrovsky-Thom-Milnor result.

\begin{theorem}\label{thm:oleinik}
Let $Q_1,\ldots,Q_N$ be quadratic polynomials with $s$ variables. The number of semi-algebraically connected components of realizable sign conditions of $Q_1,\ldots,Q_N$ on $\mathbb{R}^s$ is $O\left( (2N)^{s} \right)$. 
\end{theorem}

\begin{figure}[ht]	\centering
{\includegraphics[scale=0.3]{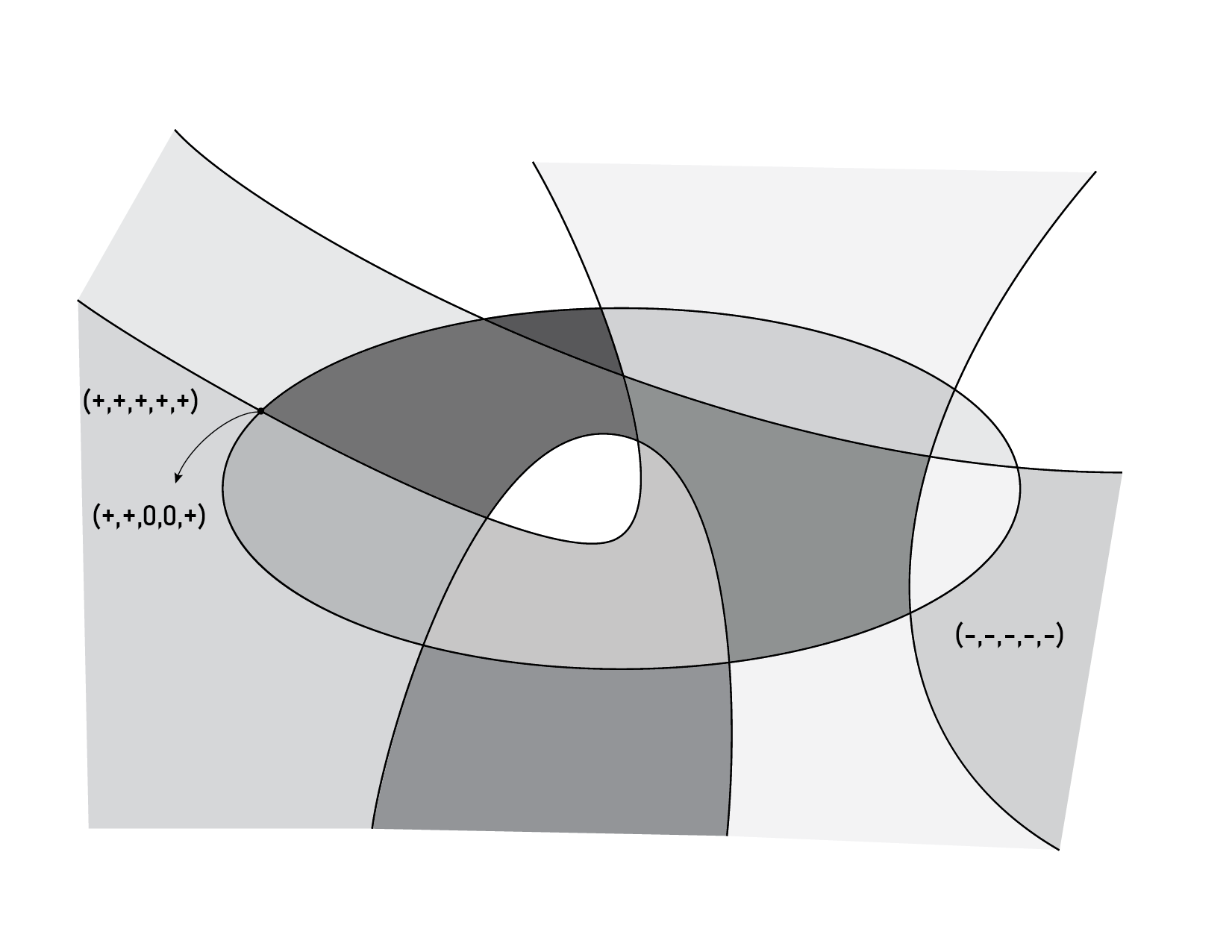}}
\caption{An arrangement of curves in the plane with the respective sign vectors.}
\end{figure}

Now suppose the input point sequences to DBA are $\gamma_1, \gamma_2, \ldots, \gamma_n$, and consider the set $\Warping^{*}_{m,k}$  of all possible optimal warping paths between a point sequence of length $m$ and a point sequence of length $k$. 
We define  quadratic polynomials $F_{\gamma_i,w}$ that encode the cost of a path $w\in \Warping^{*}_{m,k}$ on the point sequence $\gamma_i$ as follows:  
\[ F_{\gamma_i,w}(x_1,x_2,\ldots,x_k) := \sum_{(j_1,j_2)\in w} \norm{\gamma_{i,j_1} - x_{j_2} }^2 \]

To be able to compare different warping paths between input point sequences and an average point sequence $x$, we  consider the following family of quadratic polynomials.

\[ \mathcal{F} :=  \{ F_{\gamma_i,w}(x) - F_{\gamma_i,v}(x) : 1 \leq i \leq n \; , \; w, v \in \Warping^{*}_{k,m} \} \]

We make a simple observation:  If $w$ is an optimal warping path between $\gamma_i$ and $C_{\pi}$ then for all $h \in \Warping^{*}_{m,k}$ we have $F_{\gamma_i,w}(C_{\pi}) - F_{\gamma_i,h}(C_{\pi}) \leq 0$.
Now suppose the warping path $w$ between $\gamma_i$ and $C_{\alpha} \in \mathbb{R}^{dk}$  is part of the optimal assignment map between input point sequences and $C_{\alpha}$. Also suppose that DBA updates the warping path $w$ to another path $v$ after the average point sequence is updated to $C_{\beta}$. This means the following inequalities are true 
\[  F_{\gamma_i,v}(C_{\beta}) - F_{\gamma_i,w}(C_{\beta})  \leq  0 \leq F_{\gamma_i,v}(C_{\alpha}) - F_{\gamma_i,w}(C_{\alpha})  \]
Since the cost decreases at every step, both of the inequalities cannot be an equality at the same time. 
In particular, the sign of quadratic polynomial $ F_{\gamma_i,v}(x) - F_{\gamma_i,w}(x) $ is different for the values $x=C_{\alpha}$ and $x=C_{\beta}$. This implies that the  sign configurations of quadratic polynomials in $\mathcal{F}$ partition the space $\mathbb{R}^{dk}$ in such a way that two consecutive average point sequences computed by DBA are separated. 

We now argue that the algorithm visits any connected component at most once.
\begin{lemma} \label{lem:cell_visit}
 DBA does not visit a connected component of the realisable sign conditions of $ \mathcal{F}$ more than once (except for the very last step).
\end{lemma}
\begin{proof}
Assume a connected component contains $C_{\pi_i}$ and $C_{\pi_j}$ where $i<j$ and the algorithm does not converge in $C_{\pi_{j}}$. 
Because two consecutive average point sequences computed by the algorithm are separated by the realisable sign conditions of $\mathcal{F}$, we have $i+1<j$.  Since the cost decreases monotonically at every step of the algorithm, we have
\[\Psi_{\pi_i}(C_{\pi_i})>\Psi_{\pi_{i+1}}(C_{\pi_{i+1}})> \Psi_{\pi_j}(C_{\pi_j}) \]
Observe that $\pi_{i+1}$ was computed based on $C_{\pi_i}$ and is thus an  optimal assignment map for $C_{\pi_{i}}$.
Within every connected component, the set of optimal warping paths remains unchanged. Therefore, $\pi_{i+1}$ is also optimal  for $C_{\pi_j}$. Since $C_{\pi_{i+1}}$ is optimal for $\pi_{i+1}$, we have
\[\Psi_{\pi_j}(C_{\pi_j})\geq\Psi_{\pi_{i+1}}(C_{\pi_j})\geq\Psi_{\pi_{i+1}}(C_{\pi_{i+1}}).\]
This contradicts the statement that the cost decreases at every step.
\end{proof}
 
Lemma~\ref{lem:cell_visit} implies that the number of steps of DBA algorithm (except the very last step) is upper bounded by the number of visited regions in the sign realization of the family $\mathcal{F}$. Using the classical Oleinik-Petrovsky-Thom-Milnor result (Theorem~\ref{thm:oleinik}), we can give an upper bound on the number of regions in $\mathbb{R}^{dk}$ determined by signs of the quadratic forms $\mathcal{F}$.

\begin{lemma} \label{arrangement}
 The number of connected components of the realizable sign conditions of $ \mathcal{F}$ is $O((2n \abs{\Warping^{*}_{m,k}}^2)^{dk})$.
\end{lemma}

\begin{observation}\label{obs:countingwarpingpaths}
\[\abs{\Warping^{*}_{m,k}} \leq \binom{m+k-2}{m-1} \leq \min \{ m^{k-1} , k^{m-1} \}. \]
\end{observation}

\begin{proof}
Indeed, the number of paths on an $m \times k$ grid that can take steps $(i,j) \rightarrow (i+1,j)$ or $(i,j) \rightarrow (i,j+1)$ is given by $\binom{m+k-2}{m-1}$, since any such path takes $m+k-2$ steps. However, our warping paths may have diagonal steps of the form $(i,j) \rightarrow (i+1,j+1)$. Note that such a diagonal step corresponds to 'omitting' an assignment term between two points from the total sum and therefore always leads to a decrease in the cost. Thus, allowing diagonal steps does not increase the total number of realisable warping paths. 
\end{proof}

Now, the bound in Lemma~\ref{arrangement} together with Observation~\ref{obs:countingwarpingpaths} immediately gives an upper bound on the number of steps of DBA leading to the following theorem.

\begin{theorem}\label{thm:worstcasebound}
The number of iterations until DBA converges is in $ O ( (2n)^{dk} \binom{m+k-2}{m-1}^{2dk} )$. 
\end{theorem}

\subsection{Upper bound based on geometric properties of the input data} 
 We denote the set of valid assignment maps from $n$ input point sequences of length $m$ to a point sequence of length $k$ with $\assignment_{n,m,k}$. 
It is
$\abs{\assignment_{n,m,k}} \leq m^{kn} $.
For an assignment map $\pi \in \assignment_{n,m,k}$ and a point sequence $x=(x_1,x_2,\ldots,x_k)$ with $x_i \in \mathbb{R}^d$, we can rewrite the total warping distance as
\[ \Psi_{\pi}(x) := \sum_{i=1}^k \sum_{y \in S_i(\pi)} \norm{y-x_i}^2 \]
Our first observation is that for all $x$ we have $\Psi_{\pi}(x) \geq \Psi_{\pi}(C_{\pi})$.  We would like to express $\Psi_{\pi}(x)$ in a way that resembles an inertia. 
We set $I_{\pi}:=\sum_{i=1}^k \sum_{y \in S_{i}(\pi)} (   \norm{y}^2 - \norm{c_i(\pi)}^2 )$. Observe that $I_{\pi}= \Psi_{\pi}(C_{\pi})$. We see $I_{\pi}$ as the inertia of $\pi$, and we have
\[ \Psi_{\pi}(x)= I_{\pi} + \sum_{i=1}^k \abs{S_i(\pi)} \; \norm{x_i - c_i(\pi)}^2\]
We use the following geometric properties of the input:\\
\textbf{Normalization Property:} Let $B>0$. For any vector $y$ in any input point sequence we have $\norm{y}^2 \leq B$. \\
\textbf{Separation Property:} Let $\varepsilon>0$. For any two different assignment maps $\alpha, \beta \in \assignment_{n,m,k}$ we have
\[ \norm{C_{\alpha}-C_{\beta}}^2 \coloneqq \sum_{i=1}^k (c_i(\alpha)-c_i(\beta))^2 \geq \varepsilon \]
With the help of the following lemma, we can derive an upper bound in Theorem~\ref{niceupperbound}.
\begin{lemma} \label{normalization}
 For all $\alpha \in \assignment_{n,m,k}$, it is $I_{\alpha} \leq  B n (m+k) $.
\end{lemma}
\begin{proof}
\[ I_{\alpha} = \sum_{i=1}^k \sum_{y \in S_{i}(\alpha)} (   \norm{y}^2 - \norm{c_i(\alpha)}^2 ) \leq  B \sum_{i=1}^k \abs{S_{i}(\alpha)}  \]
Note that every warping path between the average point sequence of length $k$ and an input point sequence of length $m$ consists of at most $m+k$ many steps. This means every point sequence contributes to the sum $\sum_{i=1}^k \abs{S_{i}(\alpha)}$ with at most $m+k$ elements, and hence we have
$\sum_{i=1}^k \abs{S_{i}(\alpha)} \leq n (m+k)$.
\end{proof}

\begin{theorem} \label{niceupperbound}
If the input data satisfies the normalization property with parameter $B$ and separation property with parameter $\varepsilon$, then the number of steps performed by DBA is at most
\[\frac{B(m+k)}{\varepsilon}\]
\end{theorem}

\begin{proof}
Suppose DBA has started from assignment map $\pi_0$. If DBA takes a step from some assignment map $\alpha$ to some assignment map $\beta$ this means
\begin{align*}
    I_{\alpha} &= \Psi_{\alpha}(C_{\alpha})  >  \Psi_{\beta}(C_{\alpha})\\ &= I_{\beta} + \sum_{i=1}^k \abs{S_i(\beta)} \; \norm{c_i(\alpha) - c_i(\beta)}^2
\end{align*}
Since $\abs{S_i(\beta)} \geq n$ for all $i$, this implies $I_{\alpha} > I_{\beta} + n \norm{C_{\alpha}-C_{\beta}}^2$. Using the seperation property of the data and Lemma \ref{normalization} , we get
\begin{equation}
\label{potential_eq}
    I_{\alpha} > I_{\beta} + n \varepsilon  \geq  I_{\beta} + \frac{\varepsilon}{B(m+k)} I_{\pi_0}
\end{equation}
Let $\pi_T$ be the assignment map after $T$ steps of DBA, then we have by (\ref{potential_eq}) that
\[I_{\pi_T}< I_{\pi_0}-T \frac{\varepsilon}{B(m+k)}I_{\pi_0}\]
\end{proof}

\section{Smoothed Analysis}\label{sec:smoothedanalysis}

Our randomness model is as follows:   An adversary specifies an instance $X'\in ([0,1]^d)^{nm}$ of $n$ point sequences $\gamma_1,\dots,\gamma_n$ of length $m$ in $[0,1]^d$, where each sequence $\gamma_i$ is given by its $m$ points $\gamma_i=(\gamma_{i,1},\dots,\gamma_{i,m})$. Then we add to each point of $X'$ an $d$-dimensional random vector with independent Gaussian coordinates of mean $0$ and standard deviation $\sigma$.  The resulting vectors form the input point sequences. We assume without loss of generality that $\sigma\leq 1$, since the case $\sigma>1$ corresponds to a scaled down instance $X'\in ([0,\frac{1}{\sigma}]^d)^{nm}$ with additive $d$-dimensional Gaussian random vectors with mean $0$ and standard deviation $1$.  We call this randomness model $m$-length sequences with $\mathcal{N}(0,\sigma)$ perturbation.

We note that the results in this section hold for a more general family of random input models (See Section 1.5 of \cite{ergursmoothed} or Section 3.1 of \cite{ergurPV}). We conduct the analysis only for Gaussian perturbation for the sake of simplicity and obtain the following theorem. 

\begin{theorem}\label{thm:mainsmooth}
Suppose $d\geq2$, then the expected number of iterations
until DBA  converges is at most \[ O \left(\frac{n^2 m^{8\frac{n}{d}+6}d^4k^6\ln(nm)^4}{\sigma^2}  \right).\]
\end{theorem}

To proof Theorem~\ref{thm:mainsmooth}, we first bound the probability that the normalization property and the separation property hold for suitable parameters. 

\begin{lemma} \label{normalization_2}
We have
 \[ \mathbb{P} \{  \max_{1 \leq i \leq n} \max_{y \in \gamma_i} \norm{y} \geq \sqrt{d}+t \sigma \sqrt{2 d \ln n m}  \} \leq e^{1-t^2} \]
\end{lemma}
\begin{proof}
By our assumptions every vector in input sequences is given by $D + Y$ where $D$ is a deterministic vector with norm at most $\sqrt{d}$ and $Y$ is a random vector with Gaussian i.i.d coordinates $\mathcal{N}(0,\sigma)$. By triangle inequality $\norm{D+Y} \leq \sqrt{d} + \norm{Y}$. Since $Y$ has Gaussian i.i.d coordinates $\mathcal{N}(0,\sigma)$, we can apply the standard tail bound  $\mathbb{P} \{ \norm{Y} \geq t \sigma \sqrt{d}  \} \leq e^{1-\frac{t^2}{2}}$.
\end{proof}

\begin{lemma} \label{sepep}
 Let $C_{\alpha}$ and $C_{\beta}$ be  average point sequences corresponding to two different assingment maps $\alpha$ and $\beta$. Then, we have
$ \mathbb{P} \{ \norm{C_{\alpha} - C_{\beta}}^2 \leq \varepsilon \} \leq \left( \frac{nm\sqrt{ \varepsilon}}{\sigma } \right)^d$.
Furthermore, the separation property with parameter $\varepsilon$ holds with probability at least 
\[  1 - m^{4n}\left( \frac{nm\sqrt{ \varepsilon}}{\sigma } \right)^d. \]
\end{lemma}
\begin{proof}
Since the instances are perturbed and the assignment maps are different, we have with probability $1$ that there exists an $i\in[k]$ such that $c_i(\alpha)\neq c_i(\beta)$. The event $\norm{C_{\alpha} - C_{\beta}}^2\leq \varepsilon$ further implies $\norm{c_i(\alpha)-c_i(\beta)}\leq \sqrt{\varepsilon}$.
We bound the probability that this event  $\norm{c_i(\alpha)-c_i(\beta)}\leq \sqrt{\varepsilon}$ occurs for the fixed $c_i(\alpha)$ and $c_i(\beta)$. 

Let $S_i(\alpha)$ and $S_i(\beta)$ denote the sets of points in $X$ that were assigned to $c_i(\alpha)$ and $c_i(\beta)$ respectively. Since $c_i(\alpha)\neq c_i(\beta)$, it immediately follows that $S_i(\alpha)\neq S_i(\beta)$. So we can fix a point $s\in S_i(\alpha) \triangle S_i(\beta) $. We let an adversary fix all points in $S_i(\alpha) \cup S_i(\beta)\setminus \{s\}$. In order for $c_i(\alpha)$ and $c_i(\beta)$ to be $\sqrt{\epsilon}$-close, we need $s$ to fall into a hyperball of radius $nm\sqrt{\epsilon}$. Because $s$ is drawn from a Gaussian distribution with standard deviation $\sigma$, this happens with probability at most $\left(\frac{nm\sqrt{\epsilon}}{\sigma}\right)^d$. So in total, we have $  \mathbb{P} \{ \norm{c_i(\alpha)-c_i(\beta)}\leq \sqrt{\varepsilon} \}
    \leq \left( \frac{nm\sqrt{ \varepsilon}}{\sigma } \right)^d$ and therefore $\mathbb{P} \{ \norm{C_{\alpha} - C_{\beta}}^2 \leq \sqrt{\varepsilon} \}
    \leq \left( \frac{nm\sqrt{ \varepsilon}}{\sigma } \right)^d$.

To prove the second claim, we apply a union bound over all possible choices for $c_i(\alpha)$ and $c_i(\beta)$. Since each $c_i(\alpha)$ and $c_i(\beta)$ is uniquely determined by its assigned points $S_i(\alpha)$ and $S_i(\beta)$ it suffices to bound these. For each input point sequence, there are at most $\binom{m}{2}$ possible choices for the set of points that get assigned to a fixed center point: This is the case since all points that get assigned to the same center point have to be consecutive. So the points that get assigned to the center point are uniquely determined by the first and the last point that gets assigned to the center point. For $n$ input point sequences all possible assignments to an arbitrary center point are therefore bounded by $\binom{m}{2}^n$. Since we choose two center points $c_i(\alpha)$ and $c_i(\beta)$, there are at most
$\binom{m}{2}^{2n}\leq m^{4n}$
possible choices for the assigned points $S_i(\alpha)$ and $S_i(\beta)$ that determine $c_i(\alpha)$ and $c_i(\beta)$.
The statement follows by applying the union bound over  possible choices for $c_i(\alpha)$ and $c_i(\beta)$.
\end{proof}

As a combination of Lemma~\ref{normalization_2} and Lemma~\ref{sepep}, we get the following Lemma.
\begin{lemma}\label{prob_steps}
Let $\gamma_1,\gamma_2,\ldots,\gamma_n$ be independent $m$-length sequences with $\mathcal{N}(0,\sigma)$ perturbation, and suppose $d \geq 2$ and $k \leq m$. Then, the  DBA algorithm implemented on the input data $\gamma_1,\gamma_2,\ldots,\gamma_n$ converges in at most 
\[  s \left(\frac{a_1 n^2 m^{8\frac{n}{d}+5}d^3k^5\ln(nm)^3}{\sigma^2}  \right) \]
steps with probability at least $1- s^{-\frac{d}{2}} - (2n)^{-dk} -  2^{-8mdk^2}$ for all $s\geq1$ where $a_1$ is constant.
\end{lemma}
\begin{proof}
In Lemma \ref{normalization_2}, we set $t= 2 \ln( (2n)^{dk} 2^{8mdk^2}) \leq 16mdk^2 \ln(4 n)$ to get that the normalization property
is fulfilled for some $B\leq a_2\sigma^2d^3k^4m^2\ln(nm)^3$ with probability at least $1- (2n)^{-dk} - 2^{-8mdk^2}$ where $a_2$ is constant. Then we use Lemma \ref{sepep} with 
$\varepsilon = \frac{\sigma^2}{sn^2m^{8\frac{n}{d}+2}}$ and apply Theorem \ref{niceupperbound}.
\end{proof}

With the help of the Lemma~\ref{prob_steps}, we are now ready to prove Theorem~\ref{thm:mainsmooth}.
\begin{proof}[Proof of Theorem~\ref{thm:mainsmooth}]
Let $X$ be the number of steps that DBA performs. By Theorem~\ref{thm:worstcasebound}, we have that $X\leq a_2(2n)^{dk}2^{8mdk^2}$ for some constant $a_2$. We set $M\coloneqq a_2(2n)^{dk}2^{8mdk^2}$ for simplicit and get
\[ \mathbb{E} [ X ] = \sum_{i=1}^{M} \mathbb{P} \{ X \; \geq  i \} \leq  K + \sum_{t=K}^{M} \mathbb{P} \{ X \; \geq t  \}  \] for any $K$.
We set $K\coloneqq a_1 n^2 m^{8\frac{n}{d}+5}d^3k^5\ln(nm)^3$. By Lemma~\ref{prob_steps}, it is \[ \mathbb{P} \{ X \; \geq sK  \}\leq s^{-\frac{d}{2}}+(2n)^{-dk}+m^{-8mdk^2}\]
for all $s\geq 1$. Therefore, we have
\[ \sum_{t=K}^{M} \mathbb{P} \{ X \geq t\}  \leq K\cdot\sum_{s=1}^{\frac{M}{K}} s^{-\frac{d}{2}} + (2n)^{-dk} +  m^{-8mdk^2} \]
Since $d \geq 2$, we have  $s^{\frac{-d}{2}} \leq \frac{1}{s}$. Moreover, $ \frac{M}{K}  \left( (2n)^{-dk} +  m^{-8mdk^2} \right) \leq 1 $. So, we have
\[ \sum_{t=K}^{M} \mathbb{P} \{ X \geq t\}  \leq K \left( 1 + \sum_{s=1}^{\frac{M}{K}} \frac{1}{s}\right) \leq K + K\ln \frac{M}{K} \]
Hence, 
\begin{align*}
    \mathbb{E} X \leq 2K + K\ln \frac{M}{K} \leq 2K +  K mdk^2 \ln\left(\frac{4a_2}{a_1}n\right)
\end{align*}

\end{proof}
\begin{remark}
Note that for the discrete case, where the positions of the points in the center point sequence are restricted to the input points, we would get an upper bound on the number of iterations which would be polynomial in $n$, instead of exponential in $n$, since for the positions of $c_i(\alpha)$ and $c_i(\beta)$ in the proof of Lemma~\ref{sepep}, there are only ${\binom{nm}{2}}$ instead of ${\binom{m}{2}}^{2n}$ possible choices.
\end{remark}

\section{Lower bound}\label{sec:lowerbound}

In this section we give a lower bound on the worst case number of iterations until DBA converges for two input point sequences of length $m\in \NN$ and an average point sequence of length $k=\Theta(m)$ in $\RR^2$. By duplicating the input point sequences this lower bound immediately yields a lower bound for $n\in \NN$ input point sequences of length $m$ and an average point sequence of length $k=\Theta(m)$ in $\RR^d$. To construct the instance for our lower bound, we directly draw from the work of Andrea Vattani on the lower bound for the worst case number of iterations of the $k$-means method \cite{vattani2011}. We take his exact construction and modify it to suit our needs by connecting input points to point sequences and scaling up the integer weights of the points by a constant factor. Here a point with integer weight resembles multiple points in the same position that are consecutive points of the corresponding point sequence. The resulting instance will lead to a lower bound that is exponential in $m$.

\subsection{Construction}
Before we give a detailed construction, we give a motivation, why we modify a $k$-means construction to show a lower bound for the number of iteration of the  DBA algorithm. The reason is that the algorithms are very similar. In fact, we can observe that DBA can behave exactly like the $k$-means algorithm and converge in the same number of iterations if in all iterations the optimal $k$-means assignment maps are also valid DBA assignment maps.

 To construct a suitable DBA instance, the idea is to connect the points of the instance in \cite{vattani2011} to point sequences such that this condition holds. This is not directly possible, so we modify the instance such that we still get the same number of iterations for both algorithms.  A challenge here is that at some steps of the algorithm some points of the average point sequence would not get mapped to any point of one of the input point sequence by the assignment map corresponding to $k$-means. This happens independent of the choice of the point sequences and implies that the respective assignment maps are not valid. Intuitively, the challenge is solved by letting the problematic points on the average sequence "steal" points from neighbouring clusters. To ensure that the general structure remains unchanged, we replace all points of the input point sequences by multiple points in the same position. In the following, we give a detailed description of the construction of the DBA instance based on the instance in \cite{vattani2011}.

\begin{figure}[ht]
\centering
\begin{tikzpicture}[scale=1.47]
delta=0.25
alpha=4.0
\coordinate (A) at (3,0);  
   \filldraw[black] (0,0) circle (0.3pt) node[anchor=south east] {$P_i$};
   \filldraw[black] (0.06,0) circle (0.15pt) node[anchor=south west] {$Q_i$};
   \filldraw[black] (1,-0.5) circle (0.5pt) node[anchor=north west] {$A_i$};
   \filldraw[black] (1,0.5) circle (0.5pt) node[anchor=west] {$B_i$};
   \filldraw[black] (1,0.70223) circle (0.7pt) node[anchor=west] {$C_i$};
   \filldraw[black] (1,1.35739) circle (1pt) node[anchor=west] {$D_i$};
   \filldraw[black] (0,1) circle (2pt) node[anchor=north west] {$E_i$};
   \path [draw=black,loosely dashed,line width=0.1pt] (0,0) circle (1cm);
   \path [draw=black,loosely dashed,line width=0.1pt] (0,0) circle (1.25cm);
   \draw[loosely dashed,line width=0.1pt] (0,0)--(-1,0);
   \tkzText[below=2pt,black](-0.80,0){$r_i$}
   \draw[loosely dashed,line width=0.1pt] (0,0)--(-0.707*1.25,-0.707*1.25);
   \tkzText[right=2pt,black](-0.707*0.9,-0.707*0.9){$R_i$}
   
   \filldraw[black] (0,0)+(A) circle (0.3pt) node[anchor=south east] {$P_i$};
   \filldraw[black] (0.06,0)+(A) circle (0.15pt) node[anchor=south west] {$Q_i$};
   \filldraw[black] (1,-0.5)+(A) circle (0.5pt) node[anchor=north west] {$A_i$};
   \filldraw[black] (1,0.5)+(A) circle (0.5pt) node[anchor=west] {$B_i$};
   \filldraw[black] (1,0.70223)+(A) circle (0.7pt) node[anchor=west] {$C_i$};
   \filldraw[black] (1,1.35739)+(A) circle (1pt) node[anchor=west] {$D_i$};
   \filldraw[black] (0,1)+(A) circle (2pt) node[anchor=north west] {$E_i$};
   \draw[blue,line width=0.2pt] (3,1)--(3+1,1.35739);
   \draw[blue,line width=0.2pt] (2.4,0.3)--(3,1);
   \draw[blue,line width=0.2pt,-stealth] (3+1,1.35739)--(3+1,-0.5)--(3.5+1,-0.1);
   \draw[purple,line width=0.4pt] (2.4,-0.2)--(3,0)--(3.06,0)--(3+1,-0.5);
   \draw[purple,line width=0.4pt,-stealth] (3+1,-0.5)--(3.5+1,-0.3);
   
   \tkzText[below=2pt,blue](2.4,0.3){$\curveone$}
   \tkzText[below=2pt,purple](2.4,-0.2){$\curvetwo$}
   \end{tikzpicture}
\caption{Schematic drawing of gadget $G_i$ ($i\geq1$) in $k$-means instance (Left) and in DBA instance  (Right).}
\label{fig:points}
\end{figure}
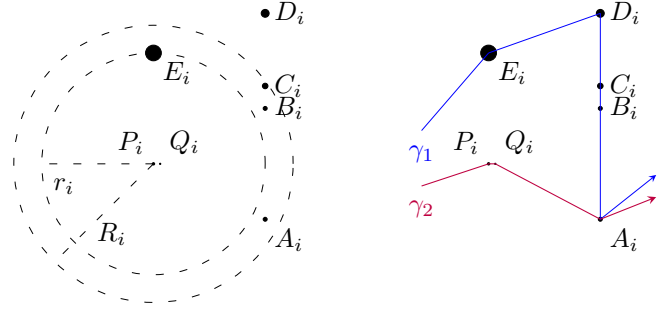

The instance $\kinst$ of Andrea Vattani consists of a sequence of gadgets $G_0,G_1,\dots,G_{\lceil\frac{k-1}{2}\rceil}$. For each $i\neq 0$, the gadget $G_i$ is given by a tuple $(\positions_i,\centers_i,r_i,R_i)$ where the set $\positions_i=\{P_i,Q_i,A_i,B_i,C_i,D_i,E_i\}\subset \RR^2$ determines the positions of the input points, $\centers_i\subset \RR^2$ determines the initial position of two centers corresponding to the gadget and $r_i>0$ and $R_i>r_i$ are the inner and outer radius of the gadget. In each position of $\positions_i$ the gadget contains a constant number of points determined by weights $w_P,w_Q,w_A,w_B,w_C,w_D,w_E$ that are independent of the index $i$ (see Table~\ref{tab:positions} for the exact values).  Note that for our use of the instance, we scale the weights by a constant factor of $100M$, where $M\in \NN$ will be determined later. The positions of each gadget are identical up to translation of the position $P_i$  and up to a scaling of the relative distance of all other positions to $P_i$. The  position can be seen in Table~\ref{tab:positions}. A $k$-means instance needs $k$ initial centers $C=\{c_1,\dots,c_k\}$. We say that a center point $c_j$ corresponds to a gadget $G_i$, if $\lceil \frac{j-1}{2} \rceil=i$. The initial position $\centers_i$ of the two centers corresponding to $G_i$ are determined as the mean of the points in an initial cluster. The two initial clusters for each gadget $G_i$ with $i\neq0$ are defined as all points in the position $A_i$ and all points in the positions $P_i,Q_i,B_i,C_i,D_i,E_i$. For an illustration of one Gadget $G_i$ of $\kinst$ with $i\neq 0$ see Figure~\ref{fig:points}. The gadget $G_0$ is given by $(\positions_0,C_0)$ where $\positions_0=C_0=\{F\}$. So all the $w_F$ points in position $F=(0,0)$ determine the position $(0,0)$ of the initial center $c_1$.

\begin{table*}[ht]
    \centering \def\arraystretch{1.3}
    \begin{tabular}{|c|c|c|}
    
         \hline 
         Weights  & Positions $\positions_i$ & Other values  \\
         \hline
         $w_P=100M$ & $P_i=S_{i-1}+(1-\eps)R_i(1,0)$ & $\eps = 10^{-6}$ \\
         $w_Q=1M$ & $Q_i=P_{i}+r_i(10^{-5},0)$ & $r_1=1$  \\
         $w_A=400M$ & $A_i=P_{i}+r_i(1,-0.5)$ & $r_i\approx 40.41608r_{i-1}$  \\
         $w_B=400M$ & $B_i=P_{i}+r_i(1,0.5)$ & $R_i=1.25r_i$  \\
         $w_C=1100M$ & $C_i\approx P_{i}+r_i(1,0.70223)$ &  $w_F=5000M$\\
         $w_D=3100M$ & $D_i\approx P_{i}+r_i(1,1.35739)$ & $F=S_0=(0,0)$ \\
         $w_E=27400M$ & $E_i=P_{i}+r_i(0,1)$ & $S_i\approx P_i+r_i(1,0,99607)$  \\
         \hline
    \end{tabular}
    \caption{ Approximate weights and positions of the points in the $k$-means instance $\kinst$ of \cite{vattani2011}.}    \label{tab:positions}
\end{table*}

To construct an instance $\dbainst$ for DBA out of the instance $\kinst$ for the $k$-means algorithm, we just connect the points of the gadgets together to form two input point sequences. The first point sequence $\curveone$ starts with half of the points in position $F$ followed by all points in the positions $E_i,D_i,C_i,B_i$ and half of the points in position $A_i$ in the given order and in the natural order $1,\dots,\lceil\frac{k-1}{2}\rceil$ of the gadgets $G_i$. The second point sequence $\curvetwo$ starts with the other half of the points in position $F$ followed by all points in the positions $P_i,Q_i$ and the other half of the points in position $A_i$ (see also Figure~\ref{fig:points}). The initial assignment map is chosen as in the $k$-means instance so that the initial average point sequence has its points $c_1,\dots,c_k$ in the same position as the centers of the $k$-means instance.

\subsection{Analysis}
\label{app:lowerbound_ana}

 For a run of the $k$-means algorithm on the instance $\kinst$, all the occurring assignment maps between points and centers in any step of the algorithm are described in  \cite{vattani2011}. Directly from the positions of the center points given by these assignment maps, we can observe the following properties for any set of center points that is present in any iteration of the $k$-means algorithm. 

\begin{observation}\label{obs:instanceprop}
Let $C=\{c_1,\dots,c_k\}$ be a set of center points present in any iteration of the $k$-means algorithm on $\kinst$. The following properties hold for $C$ and any $0\leq i \leq \lceil\frac{k-1}{2}\rceil$:
\begin{enumerate}
    \item Let $p$ be a point that lies in the gadget $G_i$. Let  $c_j\in C$ be the closest center point to $p$ (i.e. $c_j=\argmin_{c\in C}\norm{p-c}$). We have
    \[\max_{c\in\{c_{j-1},c_{j},c_{j+1}\}}\norm{p-c}\leq \alpha r_i\]
    where $\alpha>0$ is a constant that is independent of $i,p$ and the number of gadgets.
    \item Let $p$ be a point that lies in either one of the gadgets $G_{i-1},G_i,G_{i+1}$. Let further $c'=\argmin_{c\in C}\norm{p-c}$ be the closest center point and $c''=\argmin_{c\in C\setminus \{c'\}}\norm{p-c}$ bet the second closest center point. We have
    \[\norm{p-c''}-\norm{p-c'}\geq \beta r_i\]
    where $\beta>0$ is a constant that is independent of $i,p$ and the number of gadgets.
\end{enumerate}
\end{observation}

Furthermore, it is easy to check that the following method always creates a valid assignment map to the center point sequence in any iteration of DBA on the instance $\dbainst$, based on the assignment map in the same iteration of the $k$-means algorithm. 

\textit{Method for creating valid assignment map:} Take a fixed iteration of the $k$-means algorithm on $\kinst$. Let $C'=(c'_1,\dots,c'_k)$ be the sequence of centers in this iteration of the $k$-means algorithm and let $C=(c_1,\dots,c_k)$ be the points of the average point sequence in this iteration of DBA. Let $\kassign$ be the assignment map that assigns each point in $\kinst$ to its nearest neighbor in $C'$. We can interpret $\kassign$ also as an assignment map of $\dbainst$ that assigns points of $\dbainst$ that lie in the same position as points of $\kinst$ to points of the center point sequence $C$ that have the same index as the assigned centers in $C'$. Let $I$ be the set of the indices of all points in  $C$ that do not get assigned to any point of the point sequence $\curvetwo$ by this assignment map $\kassign$. Note that there are no $i,j\in I$ with $i=j+1$. For each point $j\in I$ take the last point $p_-(c_j)$ of $\curvetwo$ that got assigned to $c_{j-1}$ and the first point $p_+(c_j)$ of $\curvetwo$ that got assigned to $c_{j+1}$. Choose $p(c_j)=\argmin_{p\in\{p_-(c_j),p_+(c_j)\}}\norm{c_{j}-p}$. Replace the previous assignment of $p(c_j)$ with $c_j$. Since we scale up the instance $\dbaassign$ by a constant factor $M$, there are at least $M$ points of $Q$ at position $p(c_j)$. Only reassign the first respectively last of these points to  $c_j$. Denote the new created assignment map with $\dbaassign$.

In the following, we prove that the created valid assignment map $\dbaassign$ is also always an optimal assignment map for the instance $\dbainst$ (scaled by a suitable constant $M$) and its corresponding center point sequence. This result then directly implies that DBA on $\dbainst$ needs as many iterations as the $k$-means algorithm on $\kinst$ to converge, since the assignment map changes in each two consecutive iterations of DBA, in which the corresponding assignment map of the $k$-means algorithm changes. So DBA only converges in any iteration, in which the $k$-means algorithm converges.

\begin{theorem}\label{thm:assignopt}
Let $M>4\frac{\alpha}{\beta}+4\frac{\alpha^2}{\beta^2}$. For any iteration of DBA on the instance $X$ scaled up by the factor $M$, the valid assignment map $\dbaassign$ is an optimal assignment map.
\end{theorem}
\begin{proof}
Fix a  iteration of the $k$-means algorithm on $\kinst$. If we take the first iteration, then the center points (that the new assignment map is based on) are the same as the points of the average point sequence in the first iteration of DBA. Otherwise let $\kassign_1$ be the assignment map in the previous iteration of the $k$-means algorithm and $\dbaassign_1$ be the corresponding assignment map in the previous iteration of DBA.

We first show that the center points of $C_{\kassign_1}$ that correspond to gadget $G_i$ differ from their corresponding center points of $C_{\dbaassign_1}$ by at most $\frac{\alpha}{M}r_i$. From the construction of $\kassign_1$, we know that the assigned points to each center point $c_j(\kassign_1)$ in $\kassign_1$ and $c_j(\dbaassign_1)$ in $\dbaassign_1$ differ by at most one point. Assume that $c_j(\dbaassign_1)$ has one more point $p$. We have
\begin{align*}
    \norm{c_j(\kassign_1)-c_j(\dbaassign_1)}&=\norm{c_j(\kassign_1)-\frac{c_j(\kassign_1)\cdot|S_j(\kassign_1)|+p}{|S_j(\kassign_1)|+1})}\\
    &=\norm{\frac{c_j(\kassign_1)-p}{|S_j(\kassign)|+1}}\\
    &\leq \frac{1}{M}\norm{c_j(\kassign_1)-p}\\
    &\leq \frac{\alpha}{M}r_i
\end{align*}
Here the second to last inequality follows by $|S_j(\pi')|\geq M$ (at least one scaled up point is assigned) and the last inequality follows by the first part of Observation~\ref{obs:instanceprop}. The case that $c_j(\kassign_1)$ has one more point is analogous. By the second property of Observation~\ref{obs:instanceprop}, we know that the difference between the distances of a point from gadget $G_{i-1},G_i$ or $G_{i+1}$ to its closest center and its second closest center on $C_{\kassign_1}$ is at least $\beta r_i$. So, for $M>2\frac{\alpha}{\beta}$, it holds that 
\[2\frac{\alpha}{M}r_i<\beta r_i,\]
and it is ensured that the closest center of each point in both assignment maps $\kassign_1$ and $\dbaassign_1$ has the same index. It further holds that the difference between the distances of a point from gadget $G_{i-1},G_i$ or $G_{i+1}$ to its closest center $c'$ and its second closest center $c''$ on $C_{\dbaassign_1}$ is at least 
\begin{equation}
    \norm{c''-p}-\norm{c'-p}\geq\beta r_i-2\frac{\alpha}{M}r_i. \label{eq:Mdiffnn}
\end{equation} 
This is an important property for showing the optimality of the new assignment map. 

Let $\kassign$ be the assignment map in the fixed iteration of the $k$-means algorithm corresponding to the assignment map $\dbaassign$ in the same iteration of DBA.
 We have to show that $\dbaassign$ is an optimal assignment map between the input point sequences the average point sequence $C_{\dbaassign_1}$ generated from the assignment map of the previous iteration (or the starting average point sequence in the first iteration).
 Let $OPT=\sum_{p\in \dbainst}\min_j\norm{c_j(\dbaassign_1)-p}^2$ be the cost of the assignment map $\kassign$ between the input point sequences and $C_{\dbaassign_1}$. We know that this assignment map is optimal, if it is valid, since each point of the input point sequence has to be assigned to at least one point of the center point sequence and  $\pi'_j$ assigns each point only to its nearest neighbor. But this assignment map is not valid, since there are multiple points of $C_{\dbaassign_1}$ that do not get assigned to any point of the point sequence $\curvetwo$ in instance $\dbainst$. We denote the index set of these center points with $I$ and denote the points of $C_{\dbaassign_1}$ simply with $c_1,\dots,c_k$ instead of $c_1(\dbaassign_1),\dots, c_k(\dbaassign_1)$. 
We show that for suitable $M$ each valid assignment map results in a greater cost than $\dbaassign$, by comparing their resulting costs to $OPT$. 
 Fix an arbitrary subset $I'$ of $I$.  The cost of any assignment map that does not assign any point $c_j$ with $j\in I'$ to either $p_+(c_j)$ or $p_-(c_j)$  is greater than $OPT$ by at least
  \begin{align*}
      M\sum_{c_j\in I'}\min(\Delta_+(j),\Delta_-(j))\\+\sum_{c_j\in I\setminus I'}\norm{c_j-p(c_j)}^2-\min_t(\norm{c_t-p(c_j)}^2)
  \end{align*}
 where 
 \begin{align*}
     \Delta_+(j)&=\min_{t\neq j+1}(\norm{c_t-p_+(c_j)}^2)-\norm{c_{j+1}-p_+(c_j)}^2,\\
     \Delta_-(j)&=\min_{t\neq j-1}(\norm{c_t-p_-(c_j)}^2)-\norm{c_{j-1}-p_-(c_j)}^2.
 \end{align*}
Here we use that for each $c_j\in I$,  either all the $M$ points in position $p_+(c_j)$ or in position $p_-(c_j)$ have to get reassigned to another center point.  The assignment map $\dbaassign$ is the best assignment map that assigns all $c_j\in I$ to either $p_+(c_j)$ or $p_-(c_j)$.
The cost of $\dbaassign$ is exactly
\[\sum_{c_j\in I}\norm{c_j-p(c_j)}^2-\min_t(\norm{c_t-p(c_j)}^2)\]
greater than $OPT$. So, if  $M$ is chosen such that 
\begin{equation}
    M\min(\Delta_+(j),\Delta_-(j))> \norm{c_j-p(c_j)}^2\label{eq:Mopt}
\end{equation} then $\dbaassign$ is an optimal assignment map. 
 By the first part of Observation~\ref{obs:instanceprop}, we have  $\norm{c_j(\kassign_1)-p(c_j)}\leq \alpha r_i$. Since the center points of $C_{\kassign_1}$ that correspond to gadget $G_i$ differ from their corresponding center points of $C_{\dbaassign_1}$ by at most $\frac{\alpha}{M}r_i\leq \alpha r_i$, we further get by triangle inequality
 \[\norm{c_j-p(c_j)}\leq\norm{c_j(\kassign_1)-p(c_j)}+\norm{c_j(\kassign_1)-c_j}\leq 2\alpha r_i\]
Since $c_{j+1}$ is the closest center to $p_+(c_{j})$ and $c_{j-1}$ is the closest center to $p_-(c_{j})$, we get by Equation~(\ref{eq:Mdiffnn}) that
 \[\min(\Delta_+(j),\Delta_-(j)>(\beta-2\frac{\alpha}{M})^2 r_i^2.\]
 Here we used that for any $a>b>0$ it is $(a-b)^2>a^2-b^2$. It remains to show that 
 $M(\beta-2\frac{\alpha}{M})^2>4\alpha^2$. We have
 \begin{align*}
     M(\beta-2\frac{\alpha}{M})^2&=M\beta^2-4\beta\alpha+4\frac{\alpha^2}{M}\\
     &>M\beta^2-4\beta\alpha\\
     &>4\alpha^2
 \end{align*}
 The last inequality follows by the assumption of the theorem that $M>4\frac{\alpha}{\beta}+4\frac{\alpha^2}{\beta^2}$. So Equation~(\ref{eq:Mopt}) is fulfilled  and $\dbaassign$ is an optimal assignment map.
\end{proof}

Since $M$ can be chosen as a constant independent of $k$ the point sequences $\curveone$ and $\curvetwo$ of $\dbainst$ have a length $\Theta(k)$ each. We technically require that both point sequences have the same length. The difference in length can easily be balanced by adding an extra gadget in front of gadget $G_0$ far away from the other gadgets that contains one point of the longer point sequence and the difference in length plus one points of the smaller point sequence. Also add another center point corresponding to the gadget and initialize it at the mean of all points in the gadget. If all the points are far enough away from any other points of the instance, the corresponding center point does not change its position and the gadget does not interfere with any other gadget. By the results of \cite{vattani2011}, we know that the $k$-means algorithm needs $2^{\Omega(k)}$ iterations on the instance $\kinst$. As stated earlier, Theorem~\ref{thm:assignopt} implies that DBA on $\dbainst$ needs the same amount of iterations. We therefore achieve the following lower bound.

\begin{restatable}{theorem}{lowerbound}\label{thm:lowerbound}
Let $k\in \NN$. There is an instance of two point sequences with length $m=\Theta(k)$ in the plane such that DBA needs $2^{\Omega(k)}$ iterations to converge.
\end{restatable}

\section{Experiments on the M5 data set}\label{sec:experiments}

In this section, we present our empirical observations on the number of iterations of DBA on real world data. We run our implementation of the algorithm on real world data sets of time series and observe that in practice it runs much faster than the theoretical guarantees ensure.

\subsection{Research questions}

We are interested in how the number of iterations of the DBA algorithm depend on the complexity of the input and output for practical data sets. More specifically, we ask the following research questions.
\begin{itemize}
    \item What is the dependency of the number of iterations on the number $n$ of input point sequences?
    \item What is the dependency of the number of iterations on the length $m$ of the input point sequences?
    \item What is the dependency of the number of iterations on the length $k$ of the output center?
\end{itemize}

\subsection{Data set(s)}
To answer our research questions, we apply the DBA algorithm on the data set from the M5 Competition \cite{Makridakis2022}. Preliminary experiments on data sets from the UCR Time Series Classification Archive \cite{UCRArchive} have not shown a clear dependency on any of the quantities in question. We conjecture that that this result can be attributed to the heterogeneity of the
data sets and the relatively short length of the studied input point sequences. For a more detailed analysis of this preliminary experiments see Appendix~\ref{sec:ucr}.

The data set of the M5 Competition consists of the unit sales of products from 10 different Walmart stores in the USA. The sold number of units was tracked daily over 1942 days from 2011-01-29 to 2016-06-19 for each product aggregated for each store separately. The products can be divided into the $7$ product departments Hobbies 1-2, Foods 1-3 and Household 1-2. To create input sequences for DBA, we take for each product the time series that consists of the summed up daily sales of this product over all 10 stores. The number of input sequences in each department is given in Table~\ref{tab:departm}. In our experiments, we run DBA on sets of subsequences of such input sequences.

\begin{table}[h]
    \centering 
    \begin{tabular}{r|r|r|r|r|r|r}
         \hline 
         Hobbies 1 & Hobbies 2 & Foods 1 & Foods 2 & Foods 3 & Household 1 & Household 2 \\ 
         \hline
         416 & 149 & 216 & 398 & 823 & 532 & 515 \\
         \hline
    \end{tabular}
    \caption{Number of input sequences (different products) per department.} 
    \label{tab:departm}
\end{table}

\subsection{Setup of the experiments}
To perform experiments, DBA was implemented in C++. The implementation is based on the python implementation of Fran{\c c}ois Petitjean \cite{Petitjean2014git, petitjean2011} and can be found at \cite{git2022}. In the following, we describe the setup of the experiments.
All of our experiments use the same random initialization method to select the first center point sequence.

\subsubsection{Random initialization}
\label{sec:randin}
We initialize the first center point sequence from a random valid assignment of the input sequences. The valid assignment is created by a combination of random walks consisting of one walk per input sequence. Such a random walk starts with assigning the first point of the input series to the first point of the center point sequence and then chooses the next assignment by either moving forward on both input and center sequence or only moving forward on one of them. Each option has probability $1/3$. If it reaches the end of one sequence it continues to only move forward on the other one. After the random assignment is computed, the center points are chosen as the arithmetic means of the assigned points from the input series.

\subsubsection{Experiment 1: Dependency on the number $n$ of input point sequences}

We run DBA on sets of subsequences of consecutive days with varying sizes. The specific sizes of the sets are taken exponentially growing as 25, 50, 100, 200, 400, 800, 1600 and 3049. To choose the time series for a specific set of size $s$, we draw a number $r$ between $1$ and $3049-s$ uniformly at random and take the time series at the positions $r,r+1,\dots,r+s$.
We perform the experiment three times for three fixed values $100$, $300$ and $500$ for the length of the subsequences.
In each run of the experiment, we run DBA ten times for each fixed size of input sets and for input subseries of the fixed length, where we draw the starting day of the input subseries uniformly at random (same starting day for all series in the same run). For each run of DBA, we track the number of iterations.

\subsubsection{Experiment 2: Dependency on the length $m$ of the input point sequences}

In this experiment, we run DBA on each product department separately. This approach leads to a very natural division of the data set that creates data sets of different sizes. To be able to analyse the dependency of the number of iterations on the length of the time series, we run DBA on subsequences of input sequences with different length. We explicitly choose for each input sequence, one subsequences of consecutive days. The values for the length of the subsequences are taken exponentially growing as $15\cdot 2^i$ for $0\leq i\leq 7$ (15, 30, 60, 120, 240, 480, 960, 1920).
We run DBA ten times for each fixed value of the length, where we draw the starting day of the input subseries uniformly at random (same starting day for all series in the same run).  For each run of DBA, we track the number of iterations.

\subsubsection{Experiment 3: Dependency on the length $k$ of the output center}
 For the same reasons as in Experiment~2, we run the experiment on different Product departments separately. For sake of a clearer presentation, we restrict ourselfs to the product departments Foods 1, Foods 2 and Foods 3. As input sequences, we take all time series that correspond the respective product department. For each department, we run DBA for several length of the output center, where we test each value of the length ten times. The values for the length of the center are taken again as $15\cdot 2^i$ for $0\leq i\leq 7$ (15, 30, 60, 120, 240, 480, 960, 1920).  For each run of DBA, we track the number of iterations.

\section{Results of the experiments}
In this section, we state the results of the described experiments. More detailed tables and box plots of the exact results can be found in Appendix~\ref{app:m5}.
The results of Experiment~1, 2 and 3 are depicted in Figure~\ref{fig:m5loglognum}, \ref{fig:m5loglog} and \ref{fig:m5loglogkm}. For Experiment 1, the relation between the average number of iterations and the number of the input point sequences is shown in Figure~\ref{fig:m5loglognum}. For Experiment 2, the relation between the average number of iterations and the length of the input point sequences is shown in Figure~\ref{fig:m5loglog} and for Experiment 3, the relation between the average number of iterations and the length of the output center is shown in Figure~\ref{fig:m5loglogkm}.
\begin{figure}[H]
\centering
    \includegraphics[width=\textwidth]{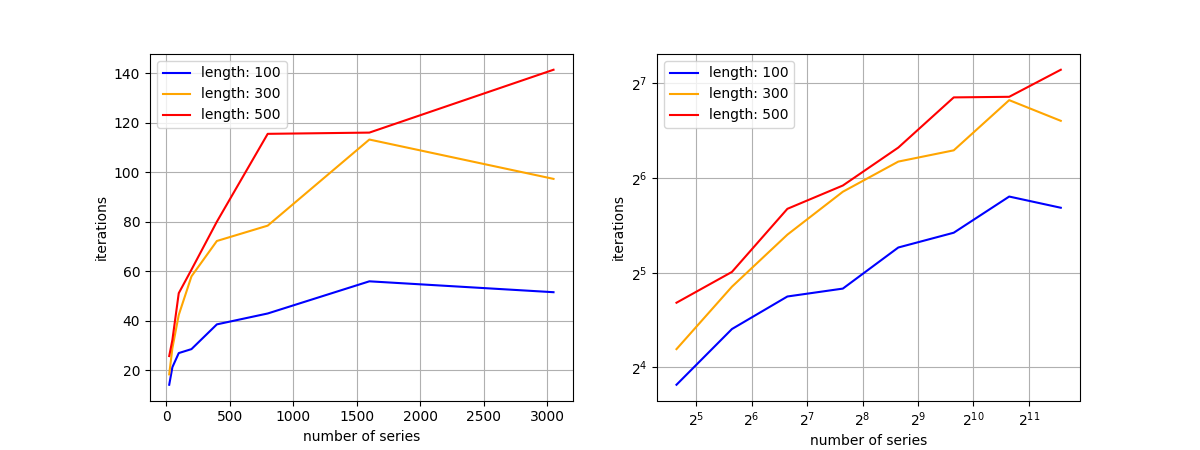}
    \caption{Depiction of the average number of iterations of the DBA algorithm with respect to the number of time series in the chosen sets. Each function graph corresponds to a fixed length of all time series in the sets. The left graphic uses normal scales and the right graphic uses logarithmic scales on both axes.}
    \label{fig:m5loglognum}
\end{figure}
\begin{figure}[H]
\centering
    \includegraphics[width=\textwidth]{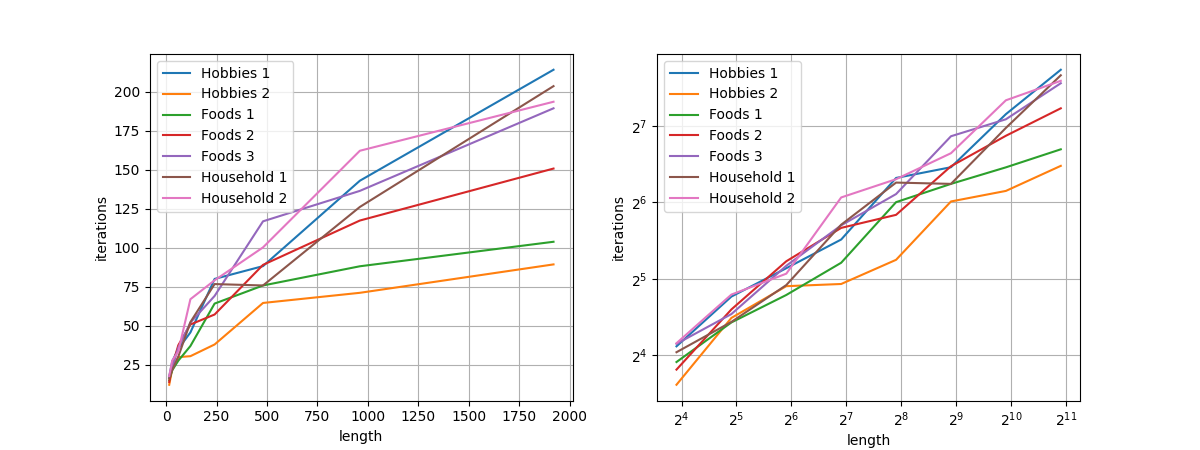}
    \caption{Depiction of the average number of iterations of the DBA algorithm with respect to the length of the series in the chosen sets. Each function graph corresponds to one product department. The left graphic uses normal scales and the right graphic uses logarithmic scales on both axes.}
    \label{fig:m5loglog}
\end{figure}
\begin{figure}[H]
\centering
    \includegraphics[width=\textwidth]{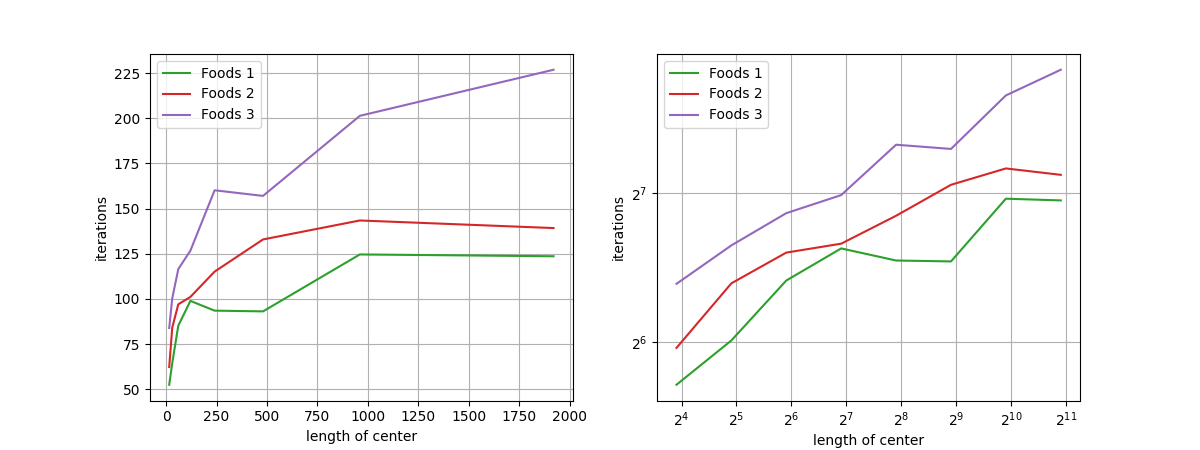}
    \caption{Depiction of the average number of iterations of the DBA algorithm with respect to the length of the center point sequence. Each function graph corresponds to one product department. The left graphic uses normal scales and the right graphic uses logarithmic scales on both axes.}
    \label{fig:m5loglogkm}
\end{figure}
As we can see in all three figures, the dependency on each parameter seems to be sublinear. Since the function graphs appear linear in the log-log plots, the dependencies seem to be polynomial for exponents smaller $1$. Estimations for the exponents which were computed with linear regression on the points in the log-log-plots are given in Table~\ref{tab:exp_size}, \ref{tab:exp_len} and \ref{tab:exp_len_cen}.

\begin{table}[H]
    \centering 
    \begin{tabular}{l||r|r|r}
         \hline 
         \hline
         Length &100 & 200 & 300 \\ 
         \hline
         Exponent &  0.27 & 0.35 & 0.36 \\
         \hline
         \hline
    \end{tabular}
    \caption{Experiment 1: Estimations for the exponents of the dependency on the size of the input point sequences for each fixed length of the input point sequences.} 
    \label{tab:exp_size}
\end{table}

\begin{table}[H]
    \centering 
    \begin{tabular}{l||r|r|r|r|r|r|r}
         \hline 
         \hline
          Dep. &Hobbies 1 & Hobbies 2 & Foods 1 & Foods 2 & Foods 3 & Household 1 & Household 2 \\ 
         \hline
         Exp. & 0.50 & 0.38 & 0.42 & 0.47 & 0.50 & 0.51 & 0.50 \\
         \hline
         \hline
    \end{tabular}
    \caption{Experiment 2: Estimations for the exponents of the dependency on the length of the output center for each Department.} 
    \label{tab:exp_len}
\end{table}

\begin{table}[H]
    \centering 
    \begin{tabular}{l||r|r|r}
         \hline 
         \hline
         Department &Foods 1 & Foods 2 & Foods 3 \\ 
         \hline
         Exponent &  0.16 & 0.16 & 0.20 \\
         \hline
         \hline
    \end{tabular}
    \caption{Experiment 3: Estimations for the exponents of the dependency on the length of the input point sequences for each Department.} 
    \label{tab:exp_len_cen}
\end{table}

\section{Conclusions}

Our experimental results show a gap between the theoretical bounds we were able to prove for semi-random input models and the number of steps DBA takes on real-world data. This suggests two directions for further research: (1) developing refined worst-case bounds under realistic input models, and (2) sharpening the obtained smoothed analysis estimates to better reflect practical performance.

The structure of the DBA heuristic is very similar to the classical $k$-means algorithm. One would naturally expect the techniques from the analysis of the $k$-means algorithm to be useful for the analysis the DBA heuristic.  This was true for the geometric ideas of Vattani in the case of proving lower bounds in Section \ref{sec:lowerbound}.  
However, surprisingly, the techniques for smoothed analysis of $k$-means were not effective in the smoothed analysis of DBA. 
The main difference seems to be following: In every step of the $k$-means algorithm, i.e. every change of assignment for a point from one center to another, the distance decreases. This gives $k$-means a very monotonic behaviour. In DBA we use dynamic time warping distance (DTW) and the behaviour is not necessarily monotonic: A change in the assignment of a point $x$ in point sequence $\gamma_1$ can \textit{increase} the distance of the point $x$ while the total DTW distance between point sequence $\gamma_1$ and the mean point sequence decreases. This behaviour seems to rule out usage of most useful techniques from \cite{roeglin2011}.
Instead, we manage to exploit the specific properties of DBA, as it requires the assignments between the points of an input sequence and the mean point sequence to respect the ordering along the mean point sequence and along the input sequence. 

In Section~\ref{sec:upperbounds}, we use  bounds on the complexity of arrangements of algebraic curves.

A similar technique was used earlier 
by Buchin et al.~\cite{buchin2022} to derive an exact exponential-time algorithm for the DTW-MEAN problem.

\bibliography{bib_short}

\appendix

\section{Experiments on the UCR Time Series Classification Archive}\label{sec:ucr}
The UCR Time Series Classification Archive contains many data sets from different areas. For our experiments, we selected the same subset of data sets as Schultz and Jain~\cite{schultz2018nonsmooth}. For each data set, we merge the training and test sets into a single data set. Each data set is divided into several classes and only contains time series of identical length. In Table~\ref{tab:ucr:chrct}, we display the name, the number $c$ of classes, the total number $n$ of time series, the average number $n/c$ of time series per cluster and the length $m$ of each data set. 

\begin{table}[H]
    \centering 
    \begin{tabular}{l||r|r|r|r||r|r||r|r}
         \hline 
         \hline 
          Dataset &  $c$ & $n$ & $n/c$ & $m$ &   $\mu_1$ & $\sigma_1^2$ &   $\mu_2$ & $\sigma_2^2$\\
         \hline
         50words & 50 & 905 & 18.10 & 270 & 50.66 & 49.56 & 45.41 & 45.86\\
         Adiac & 37 & 781 & 21.11 & 176 & 3.24 & 3.93 & 38.05 & 9.45\\
         Beef & 5 & 60 & 12.00 & 470 & 88.60 & 61.83 & 85.6 & 53.02\\
         CBF & 3 & 930 & 310.00 & 128 & 98.00 & 39.30 & 67.77 & 18.93\\
         ChlorineConcentration & 3 & 4307 & 1435.67 & 166  & 78.67 & 10.34 & 96.67 & 38.42\\
         Coffee & 2 & 56 & 28.00 & 286 & 37.50 & 8.50 & 40.15 & 6.37\\
         ECG200 & 2 & 200 & 100.00 & 96 & 65.50 & 12.50 & 53.40 & 9.53\\
         ECG5000 & 5 & 5000 & 1000.00 & 140 & 110.40 & 80.50 & 136.50 & 144.87\\
         ElectricDevices  & 7 & 16637 & 2376.71 & 96 & 50.14 & 26.44 & 57.96 & 47.15\\
         FaceAll & 14 & 2250 & 160.71 & 131 & 40.43 & 17.88 & 67.29 & 28.81\\
         FaceFour & 4 & 112 & 28.00 & 350 & 29.50 & 9.96 & 33.45 & 7.89\\
         Fish & 7 & 350 & 50.00 & 463 & 93.57 & 33.67 & 77.84 & 23.20\\
         Gun\_Point & 2 & 200 & 100.00 & 150 & 46.00 & 4.00 & 62.75 & 22.29\\
         Lighting2  & 2 & 121 & 60.50 & 637 & 80.00 & 8.00 & 47.85 & 10.36\\
         Lighting7 & 7 & 143 & 20.43 & 319 & 30.29 & 24.52 & 23.64 & 8.12\\
         OliveOil & 4 & 60 & 15.00 & 570 & 34.75 & 21.09 & 49.90 &  20.10\\
         OSULeaf & 6 & 441 & 73.50 & 427 & 164.00 & 54.99 & 126.63 & 120.17\\
         PhalangesOutlineCorrect & 2 & 2658 & 1329.00 & 80 & 57.50 & 3.50 & 192.00 & 74.70\\
         SwedishLeaf & 15 & 1125 & 75.00 & 128 & 28.34 & 11.01 & 55.82 & 15.36 \\
         synthetic\_control & 6 & 600 & 100.00 & 60 & 32.00 & 8.52 & 31.05 & 10.45 \\
         Trace & 4 & 200 & 50.00 & 275 & 29.00 & 4.06 & 54.55 & 25.66\\
         Two\_Patterns & 4 & 5000 & 1250.00 & 128 & 104.25 & 23.59 & 81.53 & 28.97 \\
         wafer & 2 & 7164 & 3582.00 & 152 & 67.5 & 4.50 & 52.05 & 7.63 \\
         yoga & 2 & 3300 & 1650.00 & 426 & 672.5 & 197.5  & 563.95 & 292.07\\
         \hline
         \hline 
    \end{tabular}
    \caption{Characteristics of the UCR TS datasets and the results for applying DBA. Here, $c$ stands for the number of classes, $n$ for the number of series, $n/c$ is therefore the average number of series in a class and $m$ is the length of each time series. The value $\mu_i$ stands for the average number of iterations and $\sigma_i^2$ stands for the variance in the number of iterations over all classes. Here the index $i=1$ stands for medoid initialization and  $i=2$ stands for random initialization.} 
    \label{tab:ucr:chrct}
\end{table}

We run the DBA algorithm on the chosen data sets for two different initialization methods (medoid and random) and analyse the dependency of the number of iterations on the number and length of the series. As input, we take each class of the data set separately and track the number of iterations needed by the DBA algorithm to converge. In the end, we compute the average number of iterations over all classes and the corresponding variances.

For the medoid initialization, we take the input point sequence as the starting center that minimizes the DTW distance to the other input point sequences. 
For the random initialization, we construct the starting center from a valid assignment of the input sequences. The valid assignment is created by a combination of random walks consisting of one walk per input sequence. Such a random walk starts with assigning the first point of the input series to the first point of the center point sequence and then chooses the next assignment by either going to the next point on the input series while staying at the same point of the center point sequence, going to then next point on the center point sequence while staying at the same point of the input series or moving to the next point on both series. Each option is taken with probability $1/3$. It can also happen, that there is only one valid option remaining, which is then always chosen. After the random assignment is computed, the center points are chosen as the arithmetic means of the assigned points from the input series. In the case of the random initialization, we run the DBA algorithm 10 times per cluster and calculate the average number of iterations over all runs on all clusters.

For each data set, we depict the tracked average number of iterations of the DBA algorithm  and the corresponding variances in Table~\ref{tab:ucr:chrct} and Figure~\ref{fig:med:num}. The figures do not suggest any clear dependency of the number of iterations on the number or length of the input point sequences for any of the two initialization methods. We conjecture that that this result can be attributed to the heterogeneity of the data sets.

\begin{figure}[H]
    \centering    \includegraphics[width=\textwidth]{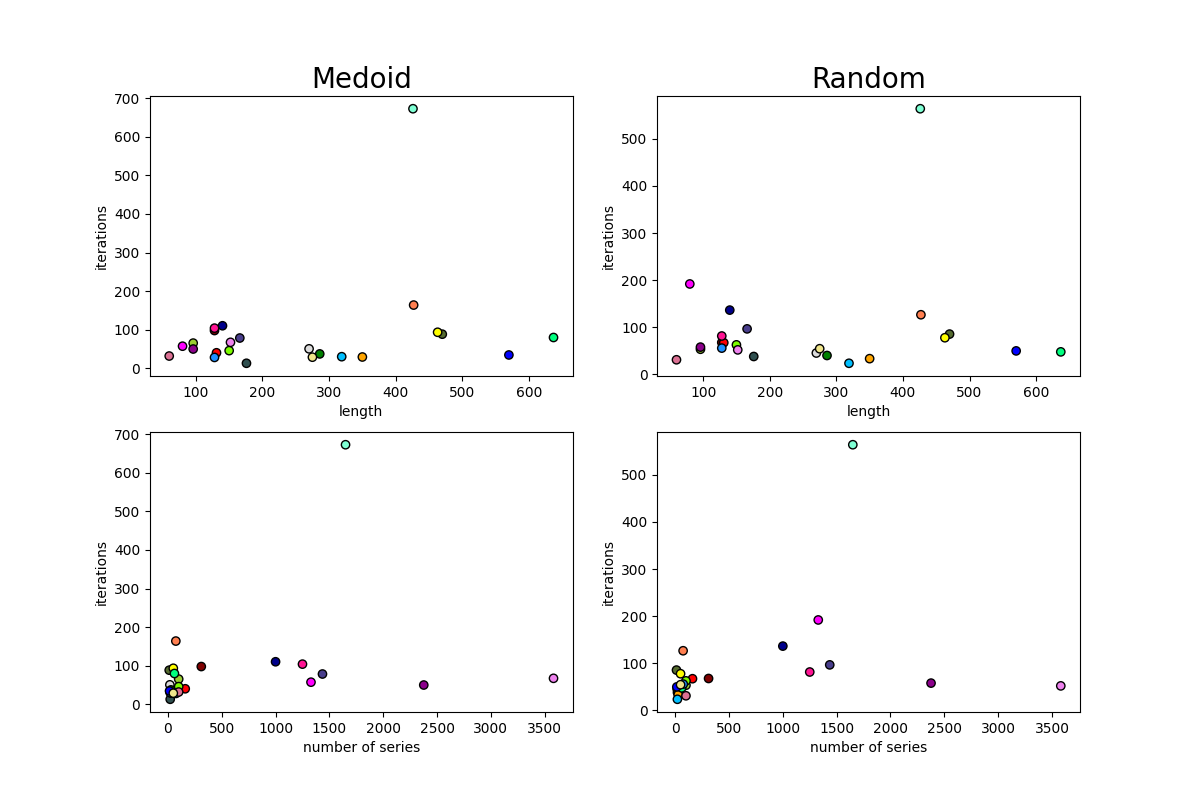}
    \caption{Depiction of the average number of iterations of the DBA algorithm with respect to the length and number of the series in the corresponding data sets. Each color corresponds to one data set of the UCR Time Series Clasification Archive. The graphics on the left side show the results for the medoid initialization and the ones on the right side show the results for the random initialization.}
    \label{fig:med:num}
\end{figure}

\section{Data of the experiments on the M5 data set}
\label{app:m5}
This chapter includes the detailed data of the experiments on the M5 data set. 
The data of Experiment~1 about the Dependency on the number of input point sequences is depicted in Table~\ref{tab:m5:num_seq}, Table~\ref{tab:m5:num_seq:var} and Figure~\ref{fig:m5var}, the data of Experiment~2 about the Dependency on the length of the input point sequences is depicted in Table~\ref{tab:m5:len}, Table~\ref{tab:m5:len:var} and Figure~\ref{fig:m5var} and the data of Experiment~3 about the Dependency on the length of the output center is depicted in Table~\ref{tab:m5:km}, Table~\ref{tab:m5:km_var} and Figure~\ref{fig:km:box}.

\begin{table}[ht]
    \centering 
    \begin{tabular}{l||r|r|r|r|r|r|r|r}
         \hline 
         \hline 
          \multicolumn{9}{c}{Mean of the number of iterations (Experiment 1)}\\
         \hline 
         \hline 
         \multicolumn{1}{r||}{\quad\quad Sequences} &  25 & 50 & 100 & 200 & 400 & 800 & 1600 & 3149\\
          Length & & & & & & & & \\
          \hline
          \hline 
         100 & 14.1& 21.2& 26.9& 28.5& 38.5& 42.9& 55.9& 51.5 \\
         300 & 18.3& 28.9& 42.3& 57.9& 72.2& 78.4& 113.2& 97.3 \\
         500 & 25.7& 32.2& 51.1& 60.6& 80& 115.5& 116& 141.4\\
         \hline 
         \hline 
    \end{tabular}
    \caption{Mean of the number of iterations with respect to the number of input point sequences. Each row corresponds to one fixed length of the input sequences}    
    \label{tab:m5:num_seq}
\end{table}

\begin{table}[ht]
    \centering 
    \begin{tabular}{l||r|r|r|r|r|r|r|r}
         \hline 
         \hline 
          \multicolumn{9}{c}{Variance of the number of iterations (Experiment 1)}\\
         \hline 
         \hline 
         \multicolumn{1}{r||}{\quad\quad Sequences} &  25 & 50 & 100 & 200 & 400 & 800 & 1600 & 3149\\
          Length & & & & & & & & \\
          \hline
          \hline 
         100 & 2.88& 4.73& 7.53& 5.41& 8.67& 11.94& 22.87& 17.25 \\
         300 & 5.64& 6.99& 8.82& 14.02& 25.72& 26.75& 57.45& 29.33 \\
         500 & 8.94& 6.90& 7.62& 17.49& 21.67& 34.80& 34.35& 44.50\\
         \hline 
         \hline 
    \end{tabular}
    \caption{Variance of the number of iterations with respect to the number of input point sequences. Each row corresponds to one fixed length of the input sequences} 
    \label{tab:m5:num_seq:var}
\end{table}

\begin{figure}[ht]
\centering
    \includegraphics[width=\textwidth]{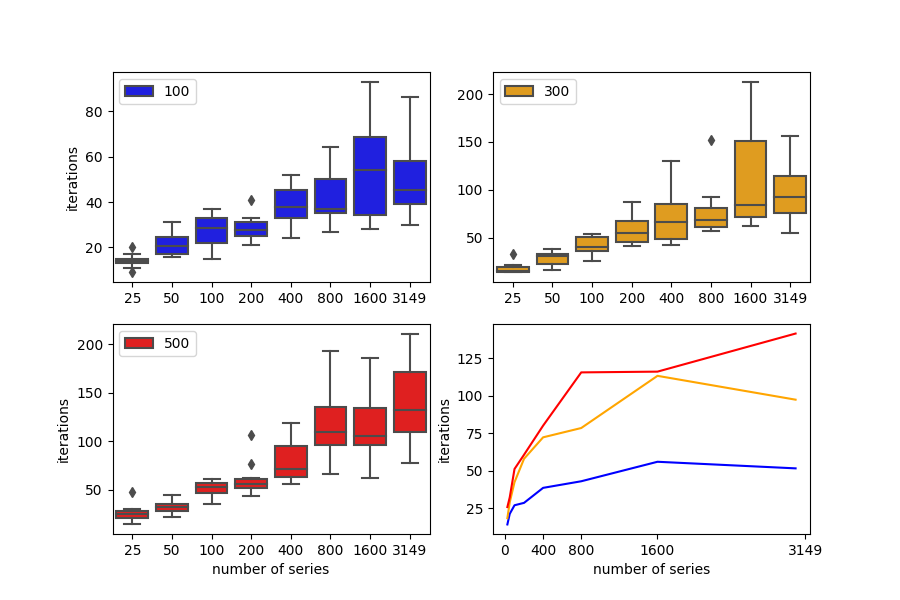}
    \caption{Experiment~1: Boxplots of the number of iterations for each department and each number of input sequences. Each graphic corresponds to one fixed length of the input sequences and has a logarithmic scale on the horizontal axis. Only the last  graphic has regular scale on the horizontal axis. It depicts the mean of the number of iterations per number of input sequences for each of the departments.} 
    \label{fig:num:box}
\end{figure}

\begin{table}[H]
    \centering 
    \begin{tabular}{l||r|r|r|r|r|r|r|r}
         \hline 
         \hline 
          \multicolumn{9}{c}{Mean of the number of iterations (Experiment 2)}\\
         \hline 
         \hline 
          \multicolumn{1}{r||}{\quad\quad\quad\quad Length} &  15 & 30 & 60 & 120 & 240 & 480 & 960 & 1920\\
           Department & & & & & & & & \\
          \hline
          \hline 
         Hobbies 1 & 17.3 &  27.2&  35.2&  45.7&  80&  88.2&  142.9&  214\\
         Hobbies 2 & 12.2 & 22.4& 29.9& 30.5& 38& 64.6& 71.1& 89.3 \\
         Foods 1 & 15& 21.5& 27.6& 37& 64.2& 75.8& 88.1& 103.8\\
         Foods 2 & 14& 24.2& 37.5& 50.8& 57.2& 89& 117.4& 150.7\\
         Foods 3 & 17.7& 23.2& 36& 52& 69.1& 116.9& 136.4& 189.3\\
         Household 1 & 16.4& 21.6& 30.2& 52.3& 76.7& 75.8& 126.1& 203.5\\
         Household 2 & 17.8& 27.7& 33.5& 67& 79.1& 100.2& 162.1& 193.5 \\
         \hline 
         \hline 
    \end{tabular}
    \caption{Mean of the number of iterations with respect to the length of the input point sequences. Each row corresponds to one product department.} 
    \label{tab:m5:len}
\end{table}

\begin{table}[p]
    \centering 
    \begin{tabular}{l||r|r|r|r|r|r|r|r}
         \hline 
         \hline 
          \multicolumn{9}{c}{Variance of the number of iterations (Experiment 2)}\\
         \hline 
         \hline 
          \multicolumn{1}{r||}{\quad\quad\quad\quad Length} &  15 & 30 & 60 & 120 & 240 & 480 & 960 & 1920\\
           Department & & & & & & & & \\
          \hline
          \hline 
         Hobbies 1 & 3.38& 11.12& 12.66& 15.47& 30.43& 20.36& 44.97& 70.87\\
         Hobbies 2 & 2.56& 9.31& 18.56& 8.09& 11.06& 12.86& 7.99& 28.43 \\
         Foods 1 & 4.47& 5.57& 10.34& 12.17& 15.07& 21.89& 10.19& 27.32\\
         Foods 2 & 3.16& 7.00& 11.81& 19.72& 14.61& 27.66& 20.19& 54.80\\
         Foods 3 & 6.42& 6.08& 11.33& 19.44& 37.63& 45.65& 41.44& 33.05\\
         Household 1 & 3.41& 3.53& 8.22& 28.78& 32.95& 19.1405& 47.66& 68.82\\
         Household 2 & 6.87& 12.1& 9.57& 30.94& 27.06& 21.50& 40.80& 51.16 \\
         \hline 
         \hline 
    \end{tabular}
    \caption{Variance of the number of iterations with respect to the length of the input point sequences. Each row corresponds to one product department} 
    \label{tab:m5:len:var}
\end{table}

\begin{figure}[p]
\centering
    \includegraphics[width=\textwidth]{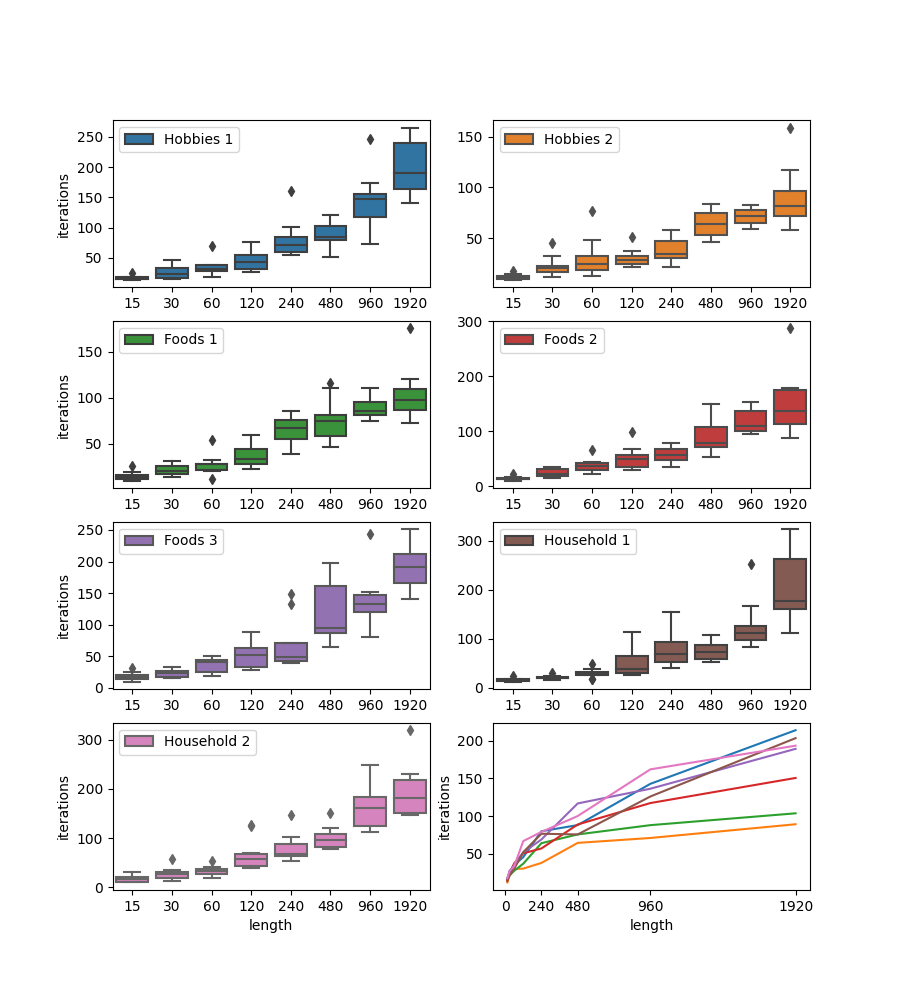}
    \caption{Experiment~2: Boxplots of the number of iterations for each department and each length of input sequences. Each graphic corresponds to one product department and has a logarithmic scale on the horizontal axis. Only the last graphic has regular scale on the horizontal axis. It depicts the mean of the number of iterations per length for each of the departments.}
    \label{fig:m5var}
\end{figure}

\begin{table}[ht]
    \centering 
    \begin{tabular}{l||r|r|r|r|r|r|r|r}
         \hline 
         \hline 
          \multicolumn{9}{c}{Mean of the number of iterations (Experiment 3)}\\
         \hline 
         \hline 
         \multicolumn{1}{r||}{\quad\quad Length of Center} &  15 & 30 & 60 & 120 & 240 & 480 & 960 & 1920\\
          Department & & & & & & & & \\
          \hline
          \hline 
         Foods 1 & 52.5& 64.5& 85.2& 98.9& 93.5& 93.1& 124.6& 123.6 \\
         Foods 2 & 62.3& 84.1& 97& 101.1& 115.1& 132.9& 143.4& 139.2 \\
         Foods 3 & 83.9& 100.3& 116.5& 126.7& 160.1& 157& 201.3& 226.8\\
         \hline 
         \hline 
    \end{tabular}
    \caption{Mean of the number of iterations with respect to the length of the output center. Each row corresponds to one product department.} 
    \label{tab:m5:km}
\end{table}

\begin{table}[ht]
    \centering 
    \begin{tabular}{l||r|r|r|r|r|r|r|r}
         \hline 
         \hline 
          \multicolumn{9}{c}{Variance of the number of iterations (Experiment 3)}\\
         \hline 
         \hline 
         \multicolumn{1}{r||}{\quad\quad Length of Center} &  15 & 30 & 60 & 120 & 240 & 480 & 960 & 1920\\
          Department & & & & & & & & \\
          \hline
          \hline 
         Foods 1 & 6.07& 11.60& 15.90& 17.47& 13.07& 14.01& 20.11& 36.69 \\
         Foods 2 & 11.1& 19.24& 17.78& 13.86& 28.68& 51.42& 40.39& 32.44 \\
         Foods 3 & 25.78& 13.33& 27.15& 25.08& 50.25& 36.90& 53.62& 66.71\\
         \hline 
         \hline 
    \end{tabular}
    \caption{Variance of the number of iterations with respect to the length of the output center. Each row corresponds to one product department} 
    \label{tab:m5:km_var}
\end{table}

\begin{figure}[ht]
\centering
    \includegraphics[width=\textwidth]{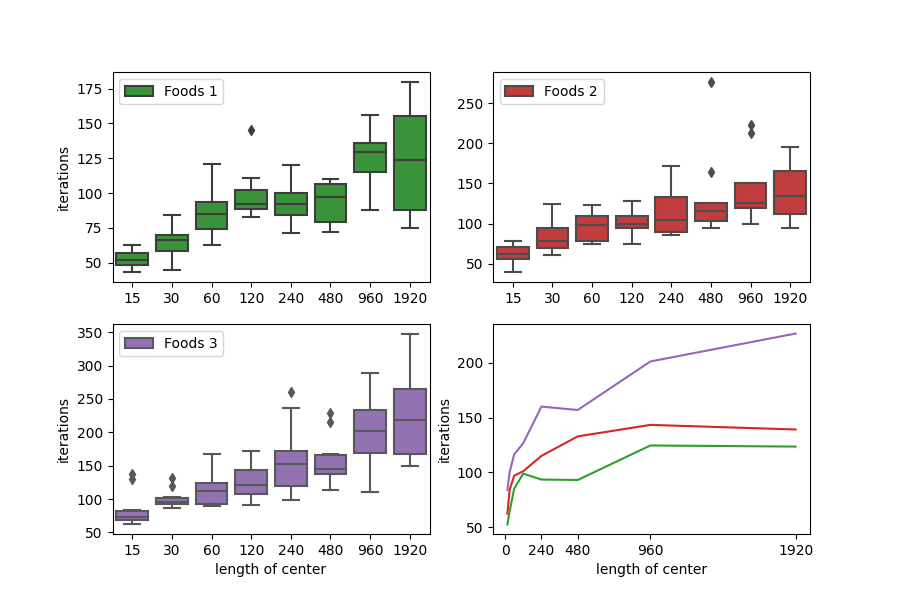}
    \caption{Experiment~3: Boxplots of the number of iterations for each department and each length of the output center. Each graphic corresponds to one product department and has a logarithmic scale on the horizontal axis. Only the last graphic has regular scale on the horizontal axis. It depicts the mean of the number of iterations per fixed length of the output center for each of the departments. }
    \label{fig:km:box}
\end{figure}

\end{document}